\documentclass[a4paper]{article}
\usepackage[all]{xy}\usepackage[latin1]{inputenc}        
\usepackage[dvips]{graphics,graphicx}
\usepackage{amsfonts,amssymb,amsmath,color,mathrsfs, amstext}
\usepackage{amsbsy, amsopn, amscd, amsxtra, amsthm,authblk}
\usepackage{enumerate,algorithmicx,algorithm}
\usepackage{algpseudocode}
\usepackage{upref}
\usepackage{color}
\usepackage{geometry}
\usepackage[displaymath]{lineno}
\usepackage{float}
\usepackage{caption}
\usepackage{yhmath}
\usepackage{booktabs}
\usepackage{subcaption}
\usepackage{multirow}
\usepackage{makecell}
\usepackage[colorlinks,
            linkcolor=red,
            anchorcolor=red,
            citecolor=red
            ]{hyperref}
\usepackage{bm}

\makeatletter
\newenvironment{breakablealgorithm}
  {
   \begin{center}
     \refstepcounter{algorithm}
     \hrule height.8pt depth0pt \kern2pt
     \renewcommand{\caption}[2][\relax]{
       {\raggedright\textbf{\ALG@name~\thealgorithm} ##2\par}%
       \ifx\relax##1\relax 
         \addcontentsline{loa}{algorithm}{\protect\numberline{\thealgorithm}##2}%
       \else 
         \addcontentsline{loa}{algorithm}{\protect\numberline{\thealgorithm}##1}%
       \fi
       \kern2pt\hrule\kern2pt
     }
  }{
     \kern2pt\hrule\relax
   \end{center}
  }
\makeatother

\numberwithin{equation}{section}

\newtheorem{theorem}{Theorem}

\newtheorem{remark}{Remark}

\numberwithin{equation}{section}

\graphicspath{{./Figs/}}

\begin{document}

\title{An Asymptotic-Preserving and Energy-Conserving Particle-In-Cell Method for Vlasov-Maxwell Equations
\footnote{Corresponding authors. Email addresses: \tt{yangzhiguo@sjtu.edu.cn}, \tt{xuzl@sjtu.edu.cn}}
}

\author[1]{Lijie Ji}
\author[*,1,4,5]{Zhiguo Yang}
\author[1]{Zhuoning Li}
\author[2]{Dong Wu}
\author[1,3,4]{Shi Jin}
\author[*,1,4,5]{Zhenli Xu}

\affil[1]{School of Mathematical Sciences, Shanghai Jiao Tong University, Shanghai 200240, China}
\affil[2]{School of Physics and Astronomy and Collaborative Innovation Center of IFSA, Shanghai Jiao Tong University, Shanghai 200240, China}
\affil[3]{Institute of Natural Sciences, Shanghai Jiao Tong University, Shanghai 200240, China}
\affil[4]{MOE-LSC, Shanghai Jiao Tong University, Shanghai 200240, China}
\affil[5]{CMA-Shanghai, Shanghai Jiao Tong University, Shanghai 200240, China}

\date{}
\maketitle

\begin{abstract}
In this paper, we develop an asymptotic-preserving and energy-conserving (APEC) Particle-In-Cell (PIC) algorithm for the Vlasov-Maxwell system.
This algorithm not only guarantees that the asymptotic limiting of the discrete scheme is a consistent and stable discretization of  the quasi-neutral limit
of the continuous model, but also preserves  Gauss's law and energy conservation at the same time, thus it is promising
to provide stable simulations of complex plasma systems even in the quasi-neutral regime.
The key ingredients for achieving these properties include the generalized Ohm's law for electric field such that the asymptotic-preserving
discretization can be achieved, and a proper decomposition of the effects of the electromagnetic fields such that a Lagrange multiplier method
can be appropriately employed for correcting the kinetic energy.
We investigate the performance of the APEC method with three benchmark tests in one dimension, including the linear Landau damping,
the bump-on-tail problem and the two-stream instability. Detailed comparisons are conducted by including the results
from the classical explicit leapfrog and the previously developed asymptotic-preserving PIC schemes. Our numerical experiments show that the proposed APEC scheme
can give accurate and stable simulations both kinetic  and quasi-neutral regimes,  demonstrating the attractive properties of the
method crossing scales. 


\end{abstract}

\section{Introduction}
Accurate simulation of collisionless plasma plays an  important role in a broad range of applications,
such as astrophysics, nuclear physics, inertial confinement fusion and material processing \cite{EVSTATIEV2013376,LANGDON1983107,DEGOND2017429}.
When classical continuum fluid models such as the magnetohydrodynamic equations fail to characterize the behavior of the plasmas under thermal or
chemical non-equilibrium conditions, the Vlasov-Maxwell (VM) system serves as a viable kinetic description under such circumstances.

There has been a longstanding interest among mathematicians and physicists in designing numerical algorithms that inherent some intrinsic properties
of the Vlasov-Maxwell equations \cite{MORRISON1980383,LEWIS1970136,2014Energy,2015Numerical,CROUSEILLES2015224,KOSHKAROV2021107866,book},
such as the preservation of conservation laws of total energy, momentum and charge \cite{chen2011energy,chen2020semi,
villasenor1992rigorous,lapenta2011particle,UMEDA2009365,SHIROTO2022111522,OCONNOR2022108345,belaouar2009asymptotically},  mimicking  the asymptotic transition
among multiple scales at the discrete level \cite{Jinreview, 2010Asymptotic,degond2017asymptotic,belaouar2009asymptotically,golse2003vlasov,
peng2002asymptotic}, or the exploration of modern exascale computer architectures for  high performance computing
 \cite{2017SMILEI}.

Despite the simplicity in its formulation, the VM system is notoriously difficult to solve numerically due to the high dimensionality
of the Vlasov equation. In real simulations, if one adopts grid-based numerical methods to solve the six-dimensional Vlasov equation,
the computational cost will be prohibitively  expensive due to ``the curse of dimensionality''. One well-established way to tackle this issue is
the Particle-In-Cell (PIC) method \cite{langdon1992enforcing,hockney2021computer,birdsall2018plasma,dawson1983particle,1962One,2008Computational,barthelme2007generalized},
in which  Newton's second law of motion for a sequence of macro particles is solved in stead of  the Vlasov equation.
The charge and current densities in the source term of the Maxwell equations are then calculated through a particle-to-grid assignment technique.
Once the electromagnetic fields are updated, they can be interpolated back to evaluate the Lorentz force to update the velocities and positions of the macro particles.

Besides the high dimensionality of the Vlasov equation, another extremely challenging issue for numerically simulating the VM system is the existence of multi-scale
transitions \cite{spaceplasma,1994Coupling}.  For instance, non-neutral and quasi-neutral regions could coexist in the domain of interest and evolve with time together.
To be more specific, there are two typical physical parameters of plasma, i.e. the Debye length and the electron plasma period \cite{chen1984introduction}.
The former characterizes the average distance between ions and electrons, and the latter represents the oscillation period of the electrons.
When both parameters are relatively small compared with the typical macroscopic scales of the system, the corresponding system is called the quasi-neutral limit.
When quasi-neutral areas exist and the kinetic description is applied, one needs to adopt very small spatial and temporal steps
in order to obtain accurate and stable simulations. This is time consuming and practically prohibitive.

In order to overcome this numerical obstacle and accelerate the calculation where non-neutral and quasi-neutral areas coexist,
the asymptotic-preserving (AP) method has been widely used in plasma physics \cite{degond2017asymptotic,2010Asymptotic,DEGOND20121917,CORDIER20125685,Degond2012,2013Second,2019High,Ameres2018Stochastic}.
In Degond {\it et al.} \cite{degond2017asymptotic}, a reformulated VM system is proposed, which unifies the models of non-neutral and quasi-neutral limit in a single set of equations
by a generalized Ohm's law, allowing for a smooth transition from one model to the other. The PIC method designed based on this reformulated VM system has proven to be stable
in both numerically resolved and under-resolved cases, which makes it possible to simulate multiscale and complex systems with much less computational cost.

However, when one approximates the probability distribution function in the reformulated VM system via macro particles,
the generalized Ohm's law and  Newton's second law for particle motion both exist in the resultant particle-Maxwell system which can be viewed as a continuous and kinetic description
of the same set of equations. This redundancy destroys the energy conservation of the reformulated particle-Maxwell system.
To the best of our knowledge, there is no asymptotic-preserving schemes for the PIC discretization that also preserve the energy conservation law.
With the help of the reformulated VM model \cite{degond2017asymptotic}, in this paper we develop an asymptotic-preserving and energy-conserving scheme that are efficient in 
 both quasi-neutral and non-neutral systems. The main idea draws inspiration from the Lagrange multiplier method proposed by Antoine {\it et al.} \cite{antoine2021scalar}.
However, a direct application of this technique will destroy the discrete Gauss' law, which makes the side effects outweigh the benefits.
In this work, we properly introduce a nonlocal Lagrange parameter at each time step to correct the kinetic energy, without destroying the physical properties,
such as the charge conservation and the asymptotic-preserving nature of the discrete algorithm. The Lagrange parameter is solved through a scalar quadratic equation
and the computational complexity of this APEC scheme is almost the same as the original AP algorithm. We perform many numerical results to demonstrate
the attractive feature of the scheme in this paper.

The organization of this paper is as follows. In Section \ref{sec:VlasovMaxwell}, we introduce the
classical VM system and the reformulated Vlasov-Maxwell system proposed in Ref. \cite{degond2017asymptotic}. In Section \ref{sec:APEC:VP}, the original AP scheme
is recast into an operator-splitting framework and the APEC scheme is proposed. Specifically, we show the detailed steps of the APEC scheme for the Vlasov-Poisson system.
The numerical tests are performed with three classical plasma problems in the electrostatic limit, including the linear Landau damping, the bump-on-tail problem
and the two-stream instability in Section \ref{sec:numerical}, together with a detailed analysis of these results. Finally, concluding remarks are made
in Section \ref{sec:conclusion}.

\section{The Vlasov-Maxwell system and its reformulated form}
\label{sec:VlasovMaxwell}

In this section, we will introduce the classical VM system and  review the reformulated VM system \cite{degond2017asymptotic} which is consistent with the quasi-neutral limit. Then, we will briefly show the PIC method used to approximate the continuous Vlasov equation and the energy-conserving properties of the related models.

\subsection{The Vlasov-Maxwell system}

Consider the collisionless plasma where the density distribution function $f_s(\bm x,\bm v,t)$ of species $s$ satisfies the Vlasov equation. The electric field $\bm{E}$ and magnetic field $\bm{B}$ satisfy the Maxwell equations where the source terms, namely, the current density $\bm{J}$ and electric charge density $\rho$, are calculated by the density distributions of all species. The whole Vlasov-Maxwell system is given by,
\begin{subequations}
\begin{align}
 &\displaystyle \partial_{t}f_s(\bm x,\bm v,t) +\bm v \cdot\nabla_{\bm x}f_s(\bm x,\bm v,t) +\frac{q_se}{m_s}(\bm E(\bm x,t)+\bm v \times \bm B(\bm x,t))\cdot\nabla_{\bm v}f_s(\bm x,\bm v,t)=0,\\
    &\displaystyle \frac{1}{c^{2}}\partial_{t}\bm E(\bm x,t)-\nabla\times \bm B(\bm x,t)=-\mu_{0} \bm J(\bm x,t),\\
    &\displaystyle \partial_{t} \bm B(\bm x,t)+\nabla\times \bm E(\bm x,t)=0,\\
    &\displaystyle \nabla\cdot \bm E(\bm x,t)=\frac{\rho(\bm x,t)}{\epsilon_{0}},\\
    &\displaystyle \nabla\cdot{\bm B(\bm x,t)}=0,
\end{align}
\end{subequations}
where $(\bm x,\bm v,t)\in \Omega_{\bm x} \times \Omega_{\bm v} \times \mathbb{R}^+ $ are  spatial, velocity and temporal variables, $m_s$ and $q_s$  are the mass and valence of species $s$, $e$ is the elementary charge, $c$ is the speed of light, $\epsilon_0$ and $\mu_0$ are the permittivity and permeability in the vacuum. The electron number density $n_s(\bm x,t)$, the electric charge density $\rho(\bm x,t)$, the current density $\bm J(\bm x,t)$ and the stress tensor $\mathcal{S}_s(\bm x,t)$ are defined as follows:
\begin{equation}\label{eq:JS}
n_s =\int_{\Omega_{\bm v}} f_sd\bm v, ~~\rho =\sum\limits_s q_se n_s, ~~\bm J=\sum\limits_{s}q_{s}e\int_{\Omega_{\bm v}}f_{s}\bm v d\bm v, ~~\mathcal{S}_s =\int_{\Omega_{\bm v}}f_s\bm v\otimes \bm vd\bm v.
\end{equation}

\subsection{The reformulated Vlasov-Maxwell system}
We briefly  describe the non-dimensionalization of physical variables and parameters by using the assumptions of magneto-hydrodynamic~(MHD) model
to derive the dimensionless form \cite{goedbloed_poedts_2004,biskamp1997nonlinear,degond2017asymptotic}. Let $x_0$ denote a typical length scale, $t_0$  a time scale, $v_0$   a velocity scale, $n_0$  a density scale, $T_0$  a temperature scale, $E_0$ and $B_0$  the electric field and magnetic field scales, respectively. Under the MHD assumptions, one has
$ex_0E_0=m_ev_0^{2}$,
$v_0B_0=E_0$, $v_0 \ll c$, and $v_{0}=v_{0,th}$, with $v_{0,th}$ being the thermal velocity, $m_e$ being the mass of electrons. Especially, one sets the dimensionless ratio $v_0/c$ to be equal to the dimensionless Debye length $\lambda =\lambda_D/x_0=\sqrt{m_e\epsilon_0 v_{0,th}^2/(e^2n_0)}/x_0$, that is, $v_0/c =\lambda$.
By consistently normalizing the physical variables and parameters based on the normalization constants given in Table \ref{tab: tab1}, the resultant dimensionless problem takes the form:
\begin{subequations}\label{eq:VlasovMaxwell}
	\begin{align}
    &\displaystyle \partial_{t}f_s+\bm v\cdot\nabla f_s+\frac{q_s}{m_s}(\bm E+\bm v\times \bm B)\cdot\nabla_{\bm v}f_s=0,\\
    &\displaystyle \lambda^{2}\partial_{t}\bm E-\nabla\times \bm B=- \bm J,\\
    &\displaystyle \partial_{t}\bm B+\nabla\times \bm E=0,\\
    &\displaystyle \lambda^{2}\nabla\cdot \bm E=\rho,\\
    &\displaystyle \nabla\cdot \bm B=0.
\end{align}
\end{subequations}
This system \eqref{eq:VlasovMaxwell} obeys the energy conservation law as described in Theorem \ref{thm1}.

 \begin{table}[H]
 	\centering
     \caption{Normalization of variables and simulation parameters} \label{tab: tab1}
     \begin{tabular}{c |c| c |c}
     \toprule[2pt]
       Variables & Normalization  & Variables  & Normalization \\
       or parameters & constant & or parameters & constant \\\midrule[1pt]
       $\bm x$   & $x_0$ & $t$ & $t_0$\\
       $\bm v$ & $v_0$ & $f_s$ & $n_0/v_0$\\
       $n_s$ & $n_0$ & $\bm J$ & $en_{0}v_{0}$\\
       $\bm E$ & $E_0$ & $\bm B$ & $B_0$\\\bottomrule[2pt]
     \end{tabular}
 \end{table}

\begin{theorem}
 	[Energy conservation law of the VM system.]
Under natural or periodic boundary conditions, the total energy of the VM system is preserved. Namely, the following relation holds,
\begin{equation}
     \frac{\text{d}}{\text{d}t}\left(\frac{1}{2}\sum_s\int_{\Omega_{\bm x}}\int_{\Omega_{\bm v}}f_s\bm v^2d\bm xd\bm v+\frac{\lambda^{2}}{2}\int_{\Omega_{\bm x}} |\bm E(\bm x)|^{2}d\bm x+\frac{1}{2}\int_{\Omega_{\bm x}}|\bm B(\bm x)|^{2}d\bm x\right)=0. \label{eq:theorem:VM}
\end{equation}
\label{thm1}
\end{theorem}

For this dimensionless form, the system is said to  be in the quasi-neutral limit if one takes the Debye length $\lambda \rightarrow 0$. In this case, the Gauss law degenerates to $\rho(\bm x,t) =0$. By a combination of Faraday's law and  Amp\`ere's law, one obtains
\begin{equation}
\nabla \times \nabla \times\bm E = -\partial_t \bm J,
\end{equation}
which is not well-posed since the irrotational part of $\bm E$ is undetermined.
 In order to deal with the degeneracy of the quasi-neutral limit, Degond {\it et al.} \cite{degond2017asymptotic} proposed to combine the generalized Ohm's law, which is obtained by taking the first order momentum of the Vlasov equation with respect to the velocity variable and incorporating with the Maxwell system to obtain a reformulated system. This new system can degrade to the quasi-neutral limit system {\it without degeneracy} when $\lambda$ takes the zero limit. Here, we adopt this reformulated model and readers can refer to \cite{degond2017asymptotic} for more details. The reformulated system reads,
\begin{subequations}\label{eq:reformulatedMaxwell}
	\begin{align}
    &\displaystyle \partial_{t}f_s+\bm v\cdot\nabla f_s+\frac{q_s}{m_s}(\bm E+\bm v\times \bm B)\cdot\nabla_{\bm v}f_s=0,\\
    &\displaystyle \lambda^{2}\partial_{t}^2 \bm E +n_s\bm E +\nabla \times (\nabla \times \bm E) = \bm J\times \bm B-\nabla \cdot \mathcal{S}, \label{eq:gol}\\
    &\displaystyle \partial_{t}\bm B+\nabla\times \bm E=0,\\
    &\displaystyle \lambda^{2}\nabla\cdot \bm E=\rho,\\
    &\displaystyle \nabla\cdot \bm B=0.
\end{align}
\end{subequations}
The reformulated VM system \eqref{eq:reformulatedMaxwell} is equivalent to the original one \eqref{eq:VlasovMaxwell}, thus it also obeys the energy conservation law \eqref{eq:theorem:VM}. Nevertheless, once the density distribution function $f_s(\bm x,\bm v,t)$ is approximated by the macro particles, the Vlasov equation will be replaced by  Newton's equations of motion for these particles and the generalized Ohm's law becomes an redundant equation in the system of the particle motion. Thus, though the original particle-Maxwell system satisfies an energy conservation law, the reformulated particle-Maxwell system may violate the  conservation law. This will be discussed below.

\subsection{The reformulated particle-Maxwell system}
In the PIC method, the distribution function $f_s$ of species $s$ is discretized as a summation of $N_{s}$ computational (or macro) particles,
\begin{equation}\label{eq:pic}
    f_{s}(\bm x,\bm v,t)=\sum_{p=1}^{N_s}w_{p}S(\bm x-\bm x_{s, p}(t))\delta(\bm v-\bm v_{s,p}(t)),
\end{equation}
where $\delta$ is the Dirac-delta function, $S$ is the shape function of spatial variable $\bm x$, and $\bm x_{s,p}$ and $\bm v_{s,p}$ are the position and velocity
of the $p$th particle of species $s$, respectively. $w_{p}$ is the weight of the macro particle and hereafter it is assumed to be a constant.
In common PIC algorithms, the shape functions are usually assumed to be B-spline functions \cite{2017SMILEI,lapenta2015kinetic},
Gaussians, cut-off cosines or polynomials \cite{2006High,2019A}, which are symmetric and compactly supported, and satisfy the
property $\int_{\Omega_{\bm x}}S(\bm x)d\bm x=1$.

Under the macro-particle approximation \eqref{eq:pic}, if one takes the first order momentum of the Vlasov equation in Eq. \eqref{eq:reformulatedMaxwell}
with respect to the velocity field and uses the properties of the shape function, the following system of Newton's equations of motion for the
macro particles \cite{birdsall2018plasma,dawson1983particle,hockney2021computer} can be obtained,
\begin{subequations}
\label{eq:motion:total}
\begin{align}
    &\displaystyle \frac{d\bm v_{s,p}}{dt}=\frac{q_{s}m_e}{m_{s}}[\bm E(\bm x_{s,p})+\bm v_{s,p}\times \bm B(\bm x_{s,p})], \\
    &\displaystyle \frac{d \bm x_{s,p}}{d t}=\bm v_{s,p},
\end{align}
\end{subequations}
for $p=1, \cdots, N_s$, where $\bm E(\bm x_{s,p})$ and $\bm B(\bm x_{s,p})$ will be calculated by the interpolation of electromagnetic fields calculated from the Maxwell equations.
Here, for simplicity we assume that there is only a species of particles in the system $s=1$ and the ions are taken into account as background,
if no ambiguity is introduced. As a consequence, the original particle-Maxwell system takes the form:
\begin{subequations}
\label{eq:motionMaxwell}
\begin{align}
&\displaystyle \frac{d\bm v_{p}}{dt}=-\bm E(\bm x_{p})-\bm v_{p}\times \bm B(\bm x_{p}),\\
    &\displaystyle \frac{d \bm x_{p}}{d t}=\bm v_{p},\\
    &\displaystyle \lambda^{2}\partial_{t}\bm E-\nabla\times \bm B=- \bm J,\\
    &\displaystyle \partial_{t}\bm B+\nabla\times \bm E=0,\\
    &\displaystyle \lambda^{2}\nabla\cdot \bm E=1-n,\\
    &\displaystyle \nabla\cdot \bm B=0,
\end{align}
\end{subequations}
with $p=1,\cdots, N$ and $n$ being the density of electrons.
It is not difficult to prove that this standard particle-Maxwell model is energy-conserving, shown in Theorem \ref{thm2}.

\begin{theorem}[Energy conservation law of the particle-Maxwell system]
Under natural or periodic boundary conditions, the total energy of the particle-Maxwell system is preserved.
\begin{equation}
     \frac{\text{d}}{\text{d}t}\left(\sum\limits_{p=1}\limits^N\frac{1}{2}w_p\bm v_{p}^{2}+\frac{\lambda^{2}}{2}\int_{\Omega_{\bm x}} |\bm E(\bm x)|^{2}d\bm x
     +\frac{1}{2}\int_{\Omega_{\bm x}}|\bm B(\bm x)|^{2}d\bm x\right)=0.
     \label{eq:theorem:motionMaxwell}
\end{equation}
\label{thm2}
\end{theorem}

However, as mentioned in the previous section, the above system is not well-posed at the quasi-neutral limit of $\lambda\rightarrow 0.$
By taking the first-order momentum of the Vlasov equation in \eqref{eq:reformulatedMaxwell} with respect to the velocity field
and using the definitions in \eqref{eq:JS}, one arrives at the generalized Ohm's law,
\begin{equation}\label{eq:ghm}
    \partial_t \bm J = \nabla \cdot \mathcal{S} +n\bm E - \bm J\times  \bm B
    \end{equation}
Combing the Faraday's law and the Amp\`ere's law together with equation \eqref{eq:motion:total}, one finally arrives at the reformulated particle-Maxwell system:
\begin{subequations}
\label{eq:motion_reformulatedMaxwell}
\begin{align}
    &\displaystyle \frac{d\bm v_{p}}{dt}=-\bm E(\bm x_{p})-\bm v_p\times \bm B(\bm x_p),\\
    &\displaystyle \frac{d \bm x_{p}}{d t}=\bm v_{p},  ~~~~~~~~ p=1,\cdots,N,\\[5pt]
    &\displaystyle \lambda^{2}\partial_{t}^2 \bm E +n\bm E +\nabla \times (\nabla \times \bm E) = \bm J\times \bm B-\nabla \cdot \mathcal{S},\\
    &\displaystyle \partial_{t}\bm B+\nabla\times \bm E=0,\\
    &\displaystyle \lambda^{2}\nabla\cdot \bm E=1-n,\\
    &\displaystyle \nabla\cdot \bm B=0.
\end{align}
\end{subequations}
Once one approximates the Vlasov equation by Newton's equations of motion, the resulting reformulated particle-Maxwell system is not equivalent to
the original particle-Maxwell system any more, and the energy conservation for this reformulated system cannot be simply proved anymore.

\begin{remark}One can conjecture the existence of the energy conservation law under these boundary conditions for finite $\lambda$ due to the equivalence
between the reformulated Vlasov-Maxwell system and the standard Vlasov-Maxwell system.
Under natural or periodic boundary conditions, it can be shown that the total energy of the reformulated particle-Maxwell system satisfies the following balance equation,
\begin{align*}
   \frac{d^{2}}{dt^{2}}\left[\sum\limits_{p=1}\limits^{N_{p}}\frac{1}{2}w_{p}{\bm v}_{p}^{2}+\frac{\lambda^{2}}{2} \int_{\Omega_{\bm x}}|{\bm E}|^{2}d\bm x\right]
   =\int_{\Omega_{\bm x}}\bm J\cdot\partial_{t}{\bm E}d{\bm x}+\lambda^{2}\int_{\Omega_{\bm x}}\left( \partial_{t}{\bm E}\right)^2 d\bm x-\int_{\Omega_{\bm x}}\left( \partial_{t}\bm B\right)^2 d\bm x.
\end{align*}
A rigorous proof of energy conservation remains an open problem.
\end{remark}

It is then reasonable to make attempt to develop a numerical scheme started from the reformulated particle-Maxwell system
and make it preserve the Gauss law and total energy at the same time, achieving both the asymptotic preserving and energy conserving. In our scheme,
the violation of Gauss law will be corrected by the Boris correction. The energy conservation will be satisfied through the Lagrange multiplier method
which is illustrated in following section.

\section{AP and APEC schemes}\label{sec:APEC:VP}

Before introducing the numerical scheme, we should emphasize that the reformulated Maxwell system integrates the generalized Ohm's law, the Faraday equation
and the Amp\`ere equation to derive a new equation Eq.\eqref{eq:gol} which avoids the degeneracy of the quasi-neutral limit  \cite{degond2017asymptotic}.
In this section, we design the APEC algorithm from the classical Maxwell equations and the generalized Ohm's law.
Then we will prove that this numerical algorithm is the discretization of the reformulated Vlasov Maxwell system \eqref{eq:motion_reformulatedMaxwell}.

\subsection{The AP scheme}
\label{sec:AP}
We briefly review the first-order AP algorithm developed in Ref. \cite{degond2017asymptotic} for the reformulated VM system and
recast it into a semi-Lagrangian operator-splitting framework such that higher-order schemes can be developed straightforwardly.

We start with the generalized Ohm's law \eqref{eq:ghm} and use the operator splitting technique to split the above equation into the following two subproblems:
     \begin{subequations}\label{eq: ohmdis}
     \begin{align}
 &    \partial_t \bm J=\nabla \cdot \mathcal{S} - \bm J\times  \bm B,  \label{eq: parupd}\\
 &  \partial_t \bm J=n\bm E.\label{eq: implJ}
 \end{align}
      \end{subequations}
By the definitions of $\bm J$ and $\mathcal{S}$ in Eq. \eqref{eq:JS}, using macro particles to approximate $f(\bm x,\bm v,t)$,
and taking integral of Eq. \eqref{eq: parupd} over the spatial domain,
one arrives at a sequence of ODEs for the macro particles
 \begin{subequations}
 \label{eq:xpstar}
 \begin{align}
 &\displaystyle \frac{d \bm x_p}{dt}=\bm v_p,\\
 &\displaystyle \frac{d \bm v_p}{dt}=-\bm v_p \times \bm B(\bm x_p).
 \end{align}
 \end{subequations}
Denote $\bm X=(\bm x_1,\cdots, \bm x_N)^T$ and $\bm V=(\bm v_1,\cdots, \bm v_N)^T$. Given the values of the $m$th time step,
one can discretize Eqs. \eqref{eq:xpstar} and Eq. \eqref{eq: implJ} by a first-order scheme as follows:
 \begin{subequations}
 \label{eq: xpnstar}
\begin{align}
        &\displaystyle \frac{{\bm x}_{p}^{m+1,\ast}-{\bm x}_{p}^{m}}{\Delta t}={\bm v}_{p}^{m+1,\ast},\\
        &\displaystyle \frac{{\bm v}_{p}^{m+1,\ast}-{\bm v}_{p}^{m}}{\Delta t}=-{\bm v}_{p}^{m}\times {\bm B}^{m}(\bm x_{p}^{m}),
\end{align}
\end{subequations}
and
\begin{equation}\label{eq: jtilde}
\frac{\widetilde {\bm J}^{m+1}- \bm J^{m+1,\ast}}{\Delta t}=n^m\widetilde {\bm E}^{m+1},
\end{equation}
where $n^m=n(\bm X^m)$ and ${\bm J}^{m+1,\ast} = \bm J({\bm X}^{m+1,\ast},\bm V^{m+1,\ast})$ are defined through Eq. \eqref{eq:JS}.

The field quantities $\bm E$ and $\bm B$ are obtained by discretizing  Amp\`ere's and Faraday's laws in a semi-implicit manner:
\begin{subequations}
\label{eq:ap:ampere}
\begin{align}
        &\displaystyle \lambda^{2}\frac{\widetilde{\bm E}^{m+1}-\bm E^{m}}{\Delta t}-\nabla\times \bm B^{m+1}=-\widetilde{\bm J}^{m+1},\\
        &\displaystyle \frac{\bm B^{m+1}-\bm B^{m}}{\Delta t}+\nabla \times\widetilde{ \bm E}^{m+1}=0,
\end{align}
\end{subequations}
which can be reformulated as the following coupled system,
    \begin{equation} \label{eq:predictor}
        \begin{aligned}
            \begin{pmatrix}
            {\lambda^{2}}/{\Delta t}+{\Delta t}n^m & -\nabla\times\\
          \nabla\times & {1}/{\Delta t}
            \end{pmatrix}\begin{pmatrix}
            \widetilde{\bm E}^{m+1}\\
           \bm B^{m+1}
            \end{pmatrix}=\begin{pmatrix}
            {\lambda^{2}}/{\Delta t} & 0\\
            0 & {1}/{\Delta t}
            \end{pmatrix}\begin{pmatrix}
          \bm  E^{m}\\
            \bm B^{m}
            \end{pmatrix}
            +\begin{pmatrix}
            -\bm J^{m+1,\ast}\\
            0
            \end{pmatrix}.
        \end{aligned}
    \end{equation}
This is a linear system, to be solved in each time step to obtain $\widetilde{\bm E}^{m+1}$ and $\bm B^{m+1}$.

In the second step, one updates the electric field $\bm E^{m+1}$ via the Boris correction \cite{boris1970relativistic} for preserving  Gauss's law.
One introduces an auxiliary variable $P$ such that
    \begin{equation}\label{eq: Boris}
    \bm E^{m+1}=\widetilde{\bm E}^{m+1}-\nabla P.
    \end{equation}
By the charge conservation law, one has
$ \partial_t n =\nabla \cdot \bm J,$
 which can be discretized as
 \begin{equation}
 \frac{{n}^{m+1}-n^m}{\Delta t}=\nabla \cdot  {\bm J}^{m+1}.
 \end{equation}
 Here, ${\bm J}^{m+1}$ is obtained by discretizing Eq. \eqref{eq: implJ} by
 \begin{equation}\label{eq:jm1}
 \frac{ {\bm J}^{m+1}- \bm J^{m+1,\ast}}{\Delta t}=n^m {\bm E}^{m+1},
 \end{equation}
where ${\bm J}^{m+1,\ast}$ can be obtained upon ${\bm X}^{m+1,\ast}$ and $\bm V^{m+1,\ast}$ once they are  given.
Note that the form of ${\bm J}^{m+1}$ is different from $\widetilde {\bm J}^{m+1}$ in Eq. \eqref{eq: jtilde}, which is crucial for the numerical algorithm to remain non-degenerate when the Debye length $\lambda$ tends to 0. Inserting Eq. \eqref{eq: Boris} into  Gauss's law $\lambda^2 \nabla \cdot \bm E^{m+1} = 1- n^{m+1},$
one arrives at a Poisson equation with variable coefficients for variable $P$:
\begin{equation*}
 -\nabla\cdot\left[\left(\frac{\lambda^{2}}{\Delta t^{2}}+n^m\right)\nabla P\right]=\frac{1-n^m}{\Delta t^{2}}-\nabla \cdot \left[\left(\frac{\lambda^{2}}{\Delta t^{2}}
 + n^m\right)\widetilde{\bm E}^{m+1} \right] -\frac{1}{\Delta t}\nabla \cdot {\bm J}^{m+1,\ast}.
\end{equation*}
Taking divergence on both sides of Eq. \eqref{eq:ap:ampere} and inserting the resultant into the above equation, one finally obtains
    \begin{equation}\label{eq: psolve}
        -\nabla\cdot\left[\left(\frac{\lambda^{2}}{\Delta t^{2}}+n^m\right)\nabla P\right]=\frac{1-n^m}{\Delta t^{2}}-\frac{\lambda^{2}}{\Delta t^{2}}\nabla\cdot \bm E^{m}.
    \end{equation}
The correction $P$ is then determined by solving this equation.

Finally, one advances the Newton's equations of motion to obtain the positions and velocities of macro particles at the $t^{m+1}$ by using the following scheme,
    \begin{subequations}
\begin{align}
    &\displaystyle \frac{\bm v_{p}^{m+1}-\bm v_{p}^{m}}{\Delta t}=-  \bm E^{m+1}(\bm x_{p}^{m})-\frac{\bm v^{m+1}_p +\bm v^{m}_p}{2}\times \bm B^{m}(\bm x_{p}^{m}), \\
     &\displaystyle \frac{\bm x_{p}^{m+1}-\bm x_{p}^{m}}{\Delta t}=\bm v_{p}^{m+1},
\end{align}
\end{subequations}
for $p=1, \cdots, N$, where the fields at the particle positions are obtained by the interpolation using the shape function.

\subsection{The APEC scheme with Lagrange multiplier}
\label{sec:AP-EC}
We now proceed to construct an APEC algorithm based on the above AP scheme. The main idea for achieving the energy conservation is the decomposion of the electric force into two contributions and introducing a Lagrange multiplier to adjust the kinetic energy in the current step so that the total energy is conserved. The corrected velocities in the current step will then influence the computation of the electromagnetic fields.

The first step for the APEC scheme resembles that of the AP scheme in the previous section, i.e.  $\widetilde {\bm J}^{m+1}$ is computed using Eqs. \eqref{eq: xpnstar}-\eqref{eq: jtilde}.

We then solve Eq. \eqref{eq:predictor}. To be convenient, we split the electric and magnetic fields into two contributions,  $\widetilde {\bm E}^{m+1}=\widetilde {\bm E}^{m+1}_1+\widetilde {\bm E}_2^{m+1}$ and $ {\bm B }^{m+1}= {\bm B}^{m+1}_1+ {\bm B}_2^{m+1}$, and they are governed by:
  \begin{equation}
        \begin{aligned}
            \begin{pmatrix}
            {\lambda^{2}}/{\Delta t}+{\Delta t}n^m & -\nabla\times\\
          \nabla\times & {1}/{\Delta t}
            \end{pmatrix}\begin{pmatrix}
            \widetilde{\bm E}_1^{m+1}\\
           \bm B^{m+1}_1
            \end{pmatrix}=\begin{pmatrix}
            {\lambda^{2}}/{\Delta t} & 0\\
            0 & {1}/{\Delta t}
            \end{pmatrix}\begin{pmatrix}
          \bm  E^{m}\\
            \bm B^{m}
            \end{pmatrix},
        \end{aligned}
        \label{eq:APEC:field1:version2}
    \end{equation}
    and
      \begin{equation}
        \begin{aligned}
            \begin{pmatrix}
            {\lambda^{2}}/{\Delta t}+{\Delta t}n^m & -\nabla\times\\
          \nabla\times & {1}/{\Delta t}
            \end{pmatrix}\begin{pmatrix}
            \widetilde{\bm E}_2^{m+1}\\
           \bm B^{m+1}_2
            \end{pmatrix}=
            \begin{pmatrix}
            -\bm J^{m+1,\ast}\\
            0
            \end{pmatrix}.
        \end{aligned}
        \label{eq:APEC:field2:version2}
    \end{equation}
After obtaining $\{\widetilde {\bm E}^{m+1}_1,\widetilde {\bm E}^{m+1}_2  \}$ and $\{ {\bm B}^{m+1}_1,{\bm B}^{m+1}_2 \}$,
one uses Boris correction Eq. \eqref{eq: psolve} to update the electric field. Here, the main difference from the previous AP scheme is
 that we have two parts of the corrected electric field $\bm E^{m+1}$,
 \begin{align*}
    \bm E^{m+1}_1 = \widetilde{\bm E}_1^{m+1}   ,~\hbox{and}~
    \bm E^{m+1}_2 = \widetilde{\bm E}_2^{m+1}-\nabla P.
\end{align*}

The position and velocity of each macro particle is also split into two contributions according to the forces induced
by $\{\widetilde {\bm E}^{m+1}_1,\widetilde {\bm E}^{m+1}_2  \}$ and $\{ {\bm B}^{m+1}_1,{\bm B}^{m+1}_2 \}$.
The discretization schemes become,
 \begin{subequations}
 \label{eq:apec:particle1}
\begin{align}
    &\displaystyle \frac{\bm v_{p,1}^{m+1}-\bm v_{p}^{m}}{\Delta t}=- \bm E^{m+1}_1(\bm x_{p}^{m})-\frac{\bm v^{m+1}_{p,1} +\bm v^{m}_{p}}{2}\times \bm B^{m}(\bm x_{p}^{m}),\\
        &\displaystyle \frac{\bm x_{p,1}^{m+1}-\bm x_{p}^{m}}{\Delta t}=\bm v_{p,1}^{m+1},
\end{align}
\end{subequations}
and
  \begin{subequations}
  \label{eq:apec:particle1_2}
\begin{align}
    &\displaystyle \frac{\bm v_{p,2}^{m+1}-\bm 0}{\Delta t}=- \bm E^{m+1}_2(\bm x_{p}^{m})-\frac{\bm v^{m+1}_{p,2} }{2}\times \bm B^{m}(\bm x_p^{m}),\\
        &\displaystyle \frac{\bm x_{p,2}^{m+1}-\bm 0}{\Delta t}=\bm v_{p,2}^{m+1}.
\end{align}
\end{subequations}

Finally, one corrects the velocities of the macro particles by introducing a Lagrange multiplier. In each step, we determine a scalar constant $\xi^{m+1}$ to
correct the velocities of all particles
\begin{equation}\label{eq: vcorr}
\bm v_{p}^{m+1} = \bm v_{p,1}^{m+1}+\xi^{m+1} \bm v_{p,2}^{m+1},
\end{equation}
for $p=1,\cdots,N$, such that the total energy is conserved. By the energy conservation law, $\xi^{m+1}$ satisfies,
\begin{equation} \label{eq:ec}
\begin{aligned}
&  \int (\lambda^2 |\bm E^{m+1}|^2+|\bm B^{m+1}|^2 ) +  \sum_p w_p\Big(  |\bm v_{p,1}^{m+1}|^2 +2 \bm v_{p,1}^{m+1}\cdot \bm v_{p,2}^{m+1}\xi^{m+1}+ |\bm v_{p,2}^{m+1}|^2 (\xi^{m+1})^2   \Big) \\
 &=  \lambda^2  \int |\bm E^{0}|^2 +  \int |\bm B^{0}|^2+  \sum_p w_p  |\bm v_{p}^{0}|^2:=2W_0,
\end{aligned}
\end{equation}
with $W_0$ being the initial total energy. This is a quadratic equation for the scalar value $\xi^{m+1}$, and the exact solution close to 1 is adopted in our
calculations. One then updates the velocity field by equation \eqref{eq: vcorr} and the electromagnetic fields and the positions of the macro particles by
$\bm E^{m+1}=\bm E^{m+1}_1 +\bm E^{m+1}_2, \bm B^{m+1} = \bm B_1^{m+1}+\bm B_2^{m+1}$ and $\bm x_p^{m+1} = \bm x_{p,1}^{m+1} +\bm x_{p,2}^{m+1}.$
Note that the  fields and  particle positions remain unchanged in order to preserve  Gauss's law.

\begin{remark} The Boris correction is equivalent to a Lagrange multiplier method to enforce Gauss's law \cite{barthelme2007generalized}, where the correction field $P$ is the multiplier.	Here in our APEC scheme, the Lagrange multiplier $\xi$ to enforce the energy conservation is a scalar constant. It satisfies a quadratic equation and can be easily determined by an explicit formula. It is noted that the Gauss law is preserved during the correction of the total energy. It is essential in order that the PIC scheme can work accurately.
\end{remark}

The steps of the APEC PIC scheme are summarized in Algorithm \ref{alg:APEC}.

\begin{breakablealgorithm}
\label{alg:APEC}
\caption{\sf APEC PIC algorithm for Vlasov-Maxwell equations}
\begin{algorithmic}[1]
\vspace{0.5ex}
\State \quad  Initialization: Given $\Delta t$, $\Delta \bm x$ and total time $T$, assign velocities and positions of all particles and magnetic field.
Calculate initial charge densities through the particle-to-grid assignment. Solve the Poisson equation to obtain
Gauss-Law-satisfying initial electric field.
\State \mbox{\textbf{While}} $t^{m+1} < T$
\State \quad Calculate $\bm J^{m+1,\ast}$ by the particle-to-grid assignment, and then solve the two splitting parts of $\widetilde{\bm E}^{m+1}$
and $\bm B^{m+1}$ by Eqs. \eqref{eq:APEC:field1:version2} and \eqref{eq:APEC:field2:version2}.
\State \quad Calculate $P$ with the Boris correction by Eq. \eqref{eq: psolve},
and then obtain the two parts, $\bm E_1^{m+1} = \widetilde{\bm E}_1^{m+1}$ and $\bm E_2^{m+1} = \widetilde{\bm E}_2^{m+1} -\nabla P$.
 \State \quad Interpolate the electromagnetic fields on macro particles to get $\bm E^{m+1}_1(\bm x_p^m)$, $\bm E^{m+1}_2(\bm x_p^m)$, $\bm B^{m+1}_1(\bm x_p^m)$ and $\bm B^{m+1}_2(\bm x_p^m)$.
 \State \quad Evolve velocities and positions of particles by Eqs. \eqref{eq:apec:particle1} and \eqref{eq:apec:particle1_2}.
\State \quad Sum up two contributions of the electromagnetic fields and particle positions. Update the velocities of macro particles by,
$\bm v_{p}^{m+1} = \bm v_{p,1}^{m+1}+\xi^{m+1} \bm v_{p,2}^{m+1},$ where parameter $\xi^{m+1}$ is determined by solving Eq. \eqref{eq:ec}.
\State \quad Calculate the charge densities $1-n^{m+1}$ through the particle-to-grid assignment.
\State \quad $m = m+1$.
\State  \mbox{\textbf{End While}}
\end{algorithmic}
\end{breakablealgorithm}

The AP particle scheme in Section \ref{sec:AP} is consistent with the reformulated particle-Maxwell system in both the non-neutral and the quasi-neutral regimes \cite{degond2017asymptotic}.
Similarly, our APEC algorithm also starts from the particle-Maxwell system and the current is calculated from the generalized Ohm's law.
Theorem \ref{thm:ap} shows that the APEC scheme is consistent with the
reformulated particle-Maxwell system with the Boris correction for both  finite $\lambda$ and at the quasi-neutral limit $\lambda\rightarrow 0$.

\begin{theorem}\label{thm:ap}
	The APEC scheme developed in Section \ref{sec:AP-EC} is consistent with the reformulated particle-Maxwell system with the Boris correction:
	\begin{subequations}\label{eq:motion_reformulatedMaxwell_boris}
		\begin{align}
		&\displaystyle \frac{d\bm v_{p}}{dt}=-\bm E(\bm x_p)-\bm v_p \times\bm B(x_p),\\
		&\displaystyle \frac{d\bm  x_{p}}{d t}=\bm v_{p},\\
		&\displaystyle \lambda^{2}\partial_{t}^2 \widetilde{\bm E} +n\widetilde{\bm E} +\nabla \times (\nabla \times \widetilde{\bm E}) = \bm J\times \bm B-\nabla \cdot \mathcal{S}, \label{eq:rmm_c}\\
		&\displaystyle \partial_{t}\bm B+\nabla\times \bm E=0,\\
		&\displaystyle -\lambda^2 \partial_t^2 \Delta P - \nabla \cdot (n\nabla  P) = -\lambda^2 \partial_t^2 \nabla \cdot \widetilde{\bm E} -\nabla^2:\mathcal{S}-\nabla \cdot (n\widetilde{\bm E})+\nabla \cdot (\bm J\times \bm B),\label{eq:rmm_e}\\
		&\displaystyle \nabla\cdot \bm B=0,\\
		&\displaystyle  \bm E = \widetilde{\bm E}-\nabla P.
		\end{align}
	\end{subequations}
with $p=1,\ldots,N$,   for both  finite $\lambda$ and  the quasi-neutral limit $\lambda\rightarrow 0$.
\end{theorem}

\begin{proof}
Since Eqs. \eqref{eq:motion_reformulatedMaxwell_boris}(df) are the same as those in the standard particle-Maxwell system,
we only need to prove that the APEC algorithm is the discretization of Eqs. \eqref{eq:motion_reformulatedMaxwell_boris}(abce).
The current $\bm J^{m+1,\ast}$ can be equivalently written as
\begin{equation} \label{eq:current}
\bm J^{m+1,\ast} = \bm J^m+\Delta t \nabla \cdot \mathcal{S}^m -\Delta t \bm J^m \times \bm B^m,
\end{equation}
which was shown by Degond {et al.} \cite{degond2017asymptotic} (the moment form therein).
This expression will be used below.
		
In the APEC scheme, the electromagnetic fields are the summation of the two contributions,
$\widetilde{\bm E}^{m+1} = \widetilde{\bm E}_1^{m+1}+\widetilde{\bm E}_2^{m+1}$, and $\bm B^{m+1}=\bm B_1^{m+1}+\bm B_2^{m+1}$.
The electromagnetic fields from Eqs. \eqref{eq:APEC:field1:version2} and \eqref{eq:APEC:field2:version2} are equal to
those solutions from Eq. \eqref{eq:ap:ampere}.
By Eq. \eqref{eq:ap:ampere}(b), one has $\bm B^{m+1} = \bm B^m -\Delta t \nabla \times \widetilde{\bm E}^{m+1}$.
Substituting this expression into the Amp\`ere equation and using Eqs. \eqref{eq: jtilde} and \eqref{eq:current},  one obtains,
	\begin{align}
	\frac{\lambda^2( \widetilde{\bm E}^{m+1} -\bm E^m)}{\Delta t^2} = \frac{\nabla \times \bm B^m - \bm J^m}{\Delta t} - \nabla \times \nabla \times \widetilde{\bm E}^{m+1}
	-n^m \widetilde{\bm E}^{m+1} -\nabla \cdot \mathcal{S}^m +\bm J^m \times \bm B^m.
	\label{eq:proposition:field}
	\end{align}
Since the Faraday equation has the following discretization,
	\begin{align}
	\frac{\lambda^2 (\bm E^m - \bm E^{m-1})}{\Delta t} = \nabla \times \bm B^m -\bm J^m,
	\label{eq:proposition:ampere}
	\end{align}
Eq. \eqref{eq:proposition:field} can be further expressed as,
	\begin{align}
	\frac{\lambda^2( \widetilde{\bm E}^{m+1} -2\bm E^m +\bm E^{m-1})}{\Delta t^2} + \nabla \times \nabla \times \widetilde{\bm E}^{m+1} +n^m\widetilde{\bm E}^{m+1} +\nabla \cdot \mathcal{S}^m -\bm J^m\times \bm B^m =0.
	\end{align}
Apparently, this equation is the discretization of Eq. \eqref{eq:rmm_c} which is the reformulated Amp\`ere law.

The Boris correction $P$ satisfies,
	\begin{align}
	-\nabla \cdot \left[\left(\frac{\lambda^2}{\Delta t^2} +n^m\right)\nabla P\right] = \frac{1-n^m}{\Delta t^2} -\nabla \cdot \left[\left(\frac{\lambda^2}{\Delta t^2}+n^m\right) \widetilde{\bm E}^{m+1}\right] -\frac{1}{\Delta t} \nabla \cdot \bm J^{m+1,\ast}.
	\end{align}
If one inserts Eq.\eqref{eq:current} into this elliptic equation, one derives,
	\begin{align}
	-\nabla \cdot \left[\left(\frac{\lambda^2}{\Delta t^2} +n^m\right)\nabla P\right] =& \frac{1-n^m-\lambda^2 \nabla \cdot \widetilde{\bm E}^{m+1}}{\Delta t^2} -\nabla \cdot \left(n^m\widetilde{\bm E}^{m+1}\right)
	\nonumber \\
	&-\frac{1}{\Delta t}\nabla \cdot \bm J^m -\nabla^2:\mathcal{S}^m +\nabla \cdot \left(\bm J^m\times \bm B^m\right).
	\label{eq:proposition:correction}
	\end{align}
The corrected electric field at the $m$th step satisfies the Gauss law and Amp\`ere equation, i.e.,
	\begin{align*}
	1-n^m &= \lambda^2 \nabla \cdot \bm E^m,\\
	-\frac{1}{\Delta t}\nabla \cdot \bm J^m &= \frac{\lambda^2}{\Delta t^2}\nabla \cdot (\bm E^m -\bm E^{m-1}).
	\end{align*}
With these two expressions, Eq.\eqref{eq:proposition:correction} can be rearranged as
	\begin{align*}
	-\nabla \cdot \left[\left(\frac{\lambda^2}{\Delta t^2} +n^m\right)\nabla P\right] =& - \frac{\lambda^2}{\Delta t^2}\left(\nabla \cdot \widetilde{\bm E}^{m+1} -2\nabla \cdot \bm E^m +\nabla \cdot \bm E^{m-1} \right)
	-\nabla \cdot \left(n^m\widetilde{\bm E}^{m+1}\right)\nonumber \\
	&-\frac{1}{\Delta t}\nabla \cdot \bm J^m -\nabla^2:\mathcal{S}^m +\nabla \cdot \left(\bm J^m\times \bm B^m\right).
	\end{align*}
This is the discretization of Eq. \eqref{eq:rmm_e}.

The particle velocities $\bm v_{p,1}^{m+1}$, $\bm v_{p,2}^{m+1}$ and  positions $\bm x_{p,1}^{m+1}$, $\bm x_{p,2}^{m+1}$ are updated with two parts of the electric fields
from Eqs. \eqref{eq:apec:particle1} and \eqref{eq:apec:particle1_2}. This splitting operation is consistent with the discretization of the un-splitting motion equations. 
The additional advantage of this splitting technique, i.e., $\bm v_p^{m+1} = \bm v_{p,1}^{m+1} +\xi^{m+1} \bm v_{p,2}^{m+1}$, is the energy conservation being preserved exactly.
Moreover, $\bm v_{p,2}^{m+1}$ and $\bm x_{p,2}^{m+1}$ in Eq. \eqref{eq:apec:particle1_2} vanish at the $\Delta t\rightarrow0$ limit, implying that a finite $\xi^{m+1}$ does not
affect the consistency of the discretization for  Newton's equations of motion.
Therefore, the whole algorithm is first-order consistent with the reformulated particle-Maxwell system with the Boris correction, and the energy is preserved exactly at the same time.
	
\end{proof}

\section{Numerical results}
\label{sec:numerical}

We perform numerical results of the APEC PIC method.
We focus on numerical solutions of the one-dimensional electrostatic model of two species with
ions being motionless treated as the uniform background and only electrons present in the model.
At the electrostatic limit,
the magnetic field is disregarded $\bm B=0$ and the Faraday equation $\partial_{t} {\bm B}+\nabla\times {\bm E}=0$ reduces
to $\nabla\times {\bm E}=0$, which indicates that the electric field is irrotational. The electric field can be solved
through the gradient of the electric potential, $\bm E=-\nabla\phi$. Thus, in 1D, the original Vlasov-Maxwell system is reduced
to the Vlasov-Poisson system,
 \begin{subequations}\label{eq:vlasov:poisson}
\begin{align}
   &\displaystyle \partial_{t}f+ {v} \partial_{x}f- {E} \partial_{v}f=0,\\
    &\displaystyle - \lambda^2 \partial_{x} \phi=1-n.
\end{align}
\end{subequations}
Correspondingly, the particle-Poisson equations read,
\begin{subequations}
\begin{align}
  &\displaystyle \frac{d x_{p}}{dt}=v_{p},\\
    &\displaystyle \frac{d v_{p}}{d t}=-E(x_{p}), p=1,\cdots, N,\\
   &\displaystyle - \lambda^2 \partial_{x}^{2} \phi=1-n.
\end{align}
\end{subequations}
With periodic or homogeneous Dirichlet boundary condition of the Poisson equation, the particle-Poisson system has the following conservation law of energy,
\begin{equation}
\frac{d}{dt}\left[\sum\limits_{p=1}\limits^{N_{p}}\frac{1}{2}w_{p}v_{p}^{2}+\frac{\lambda^{2}}{2}\int_{\Omega_{x}}E(x)^{2}dx\right]=0.
\end{equation}

The APEC PIC algorithm for the 1D Vlasov-Poisson equations can be obtained by following Algorithm \ref{alg:APEC} with
the magnetic field set to be zero. We use central finite-difference scheme to solve the Poisson equation with the space domain $[a, b]$
being equidistantly divided into $N_x$ cells. We use the fourth-order B-spline function  as the shape function \cite{2017SMILEI}.
For the enforcement of the energy conserving, the quadratic equation of $\xi$ is simplified to,
$\bar{A}\xi^{2}+\bar{B}\xi+\bar{C}=2W_0$, with coefficients
\begin{align*}
&\bar{A}=w_p \sum_p (v_{p,2}^{m+1})^2,~\bar{B} = 2\sum_p w_pv_{p,1}^{m+1} v_{p,2}^{m+1}, ~\hbox{and}, \nonumber \\
&\bar{C} = \lambda^2 \sum_i ({E}_i^{m+1})^2 \Delta x+w_p \sum_p(v_{p,1}^{m+1})^2.
\end{align*}
This equation has two roots and we take the one close to 1 as the Lagrange multiplier. If there are no real roots, we
set $\xi=1$ which reduces to the AP scheme described in Section \ref{sec:AP}. 

We perform numerical calculations for three benchmark problems including the linear Landau damping,
the bump-on-tail problem and the two-stream instability. Besides the APEC scheme, we also
do the calculation with the AP particle scheme shown in Section \ref{sec:AP} and the classical
explicit scheme \cite{birdsall2018plasma} for a comparison study.
The classical explicit scheme is built on the leapfrog scheme both the particle motion and the Maxwell equations
together with a Boris correction for the enforcement of Gauss law. However, this explicit algorithm is not exactly energy conserving,
despite that the energy variation will be small for smaller temporal and spatial steps \cite{lapenta2011particle}.
Algorithm \ref{alg:explicitscheme} displays the classical explicit PIC scheme for Vlasov-Poisson
equations.

\begin{breakablealgorithm}
\caption{\sf The classical explicit PIC scheme for Vlasov-Poisson equations}
\begin{algorithmic}[1]
\vspace{0.5ex}
\State \quad Initialization: Given $\Delta t$, $\Delta \bm x$ and total time $T$, assign velocities and positions of all particles.
Calculate initial charge densities through the particle-to-grid assignment. Solve the Poisson's equation to obtain
Gauss-Law-satisfying initial electric field.
\State \mbox{\textbf{While}} $t^{m+1} < T$
\State \quad Update  particle velocities by $ \bm v_p^{m+1/2}-\bm v_p^{m-1/2}=-\Delta t\bm E^m(\bm x_p^m)$ and
 positions by $ \bm x_p^{m+1} -\bm x_p^m=\Delta t \bm v_p^{m+1/2}$.
\State \quad Calculate current and charge densities on grid sites, $\bm J^{m+1}=\bm J(\bm X^{m+1}, \bm V^{m+1/2})$ and $n^{m+1}=n(\bm X^{m+1})$, through the
particle-to-grid assignment.
\State \quad Update the electric field by $\lambda^2(\widetilde{\bm E}^{m+1} -\bm E^m)=- \Delta t \bm J^{m+1}$ and obtain the Boris correction
$\nabla P$ through $\lambda^2 \Delta P=\lambda^2 \nabla \cdot \widetilde{\bm E}^{m+1}-(1-n^{m+1})$.
\State \quad Correct the electric field, $\bm E^{m+1} =\widetilde{\bm E}^{m+1}-\nabla P$.
\State \quad $m=m+1$.
\State  \mbox{\textbf{End While}}
\end{algorithmic}
\label{alg:explicitscheme}
\end{breakablealgorithm}

\subsection{Landau damping}

To investigate our APEC algorithm, we first test the performance with a classical experiment in plasma physics, the linear Landau damping \cite{degond2017asymptotic,chen2011energy,
ho2018physics,belaouar2009asymptotically}.
The uniformly distributed electrons are slightly disturbed with an initial distribution function,
\begin{equation}
     f_{0}(x,v)=\left[1+\alpha\cos\left(\frac{x}{2}\right)\right]\frac{1}{\sqrt{2\pi}}e^{-\frac{v^{2}}{2}},
\end{equation}
where $\alpha$ is a small constant and we take $\alpha=0.05$ in the test. Ions are motionless and form a uniform background.
The space domain is  $[0,4\pi]$ and the dimensionless parameter takes $\lambda=1$. The Poisson equation is solved with the periodic boundary condition.
Figure \ref{fig:landau:p1}(a) displays the results of the classical explicit, the AP and the APEC algorithms for the resolved case with $N_x=250$, $\Delta t = 0.05$ and $N = 1\times 10^6$.  All PIC algorithms capture the correct damping rate of the perturbed plasma. The results of APEC and explicit algorithms give almost the same rate as the theoretical rate (the dotted line) during the damping period, while the AP algorithm is slightly overdamping. Figure \ref{fig:landau:p1}(b) indicates the change of the total energy ($\Delta W= W-W_0$) during the simulations and in this case $W_0=6.3176$. One can observe that the energy of the proposed APEC algorithm is exactly preserved, but the results of AP and the classical explicit methods monotonically decrease with time with different dissipative rates.

\begin{figure}[H]
\centering
\includegraphics[width=0.45\linewidth]{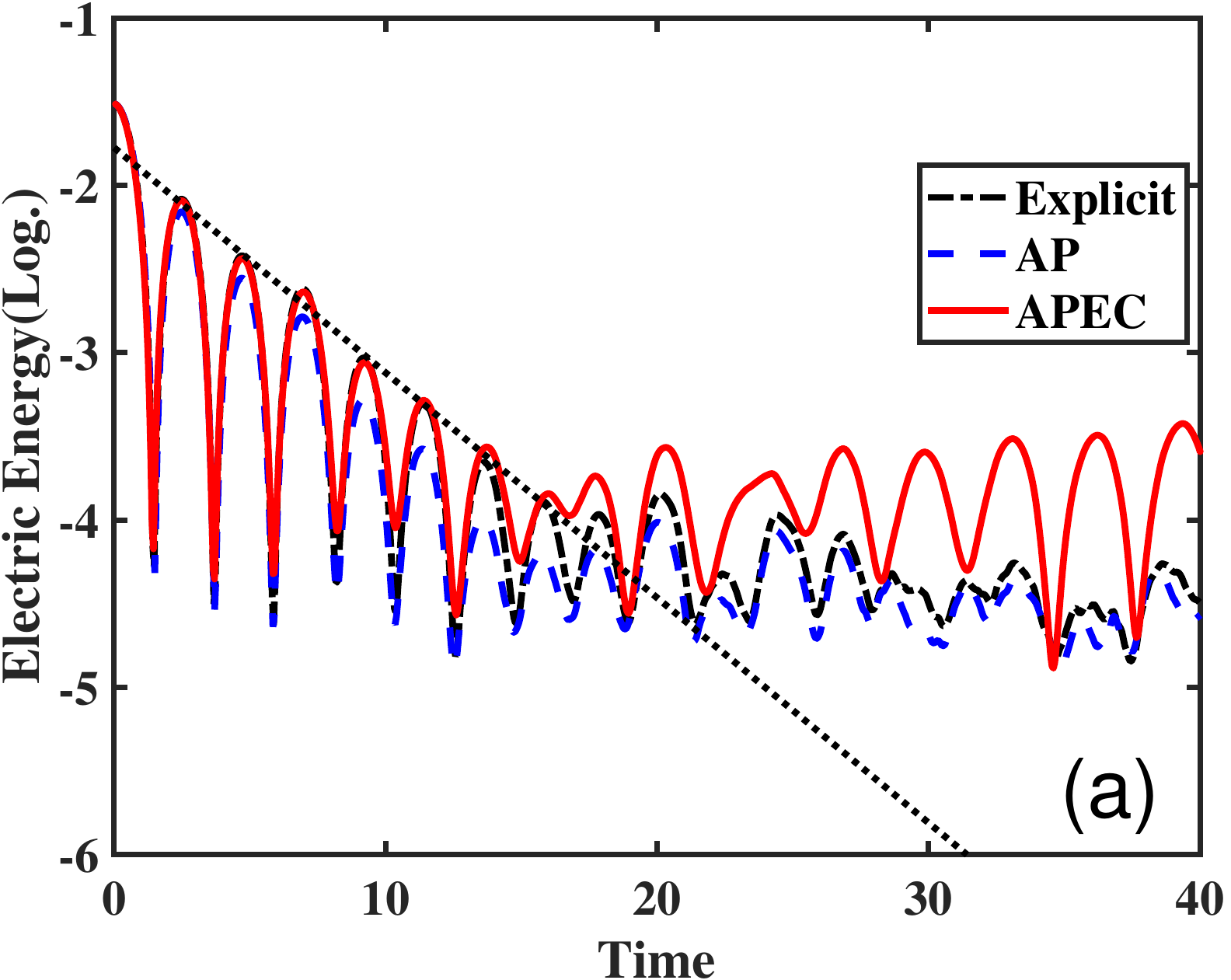}\hspace{2mm}
\includegraphics[width=0.45\linewidth]{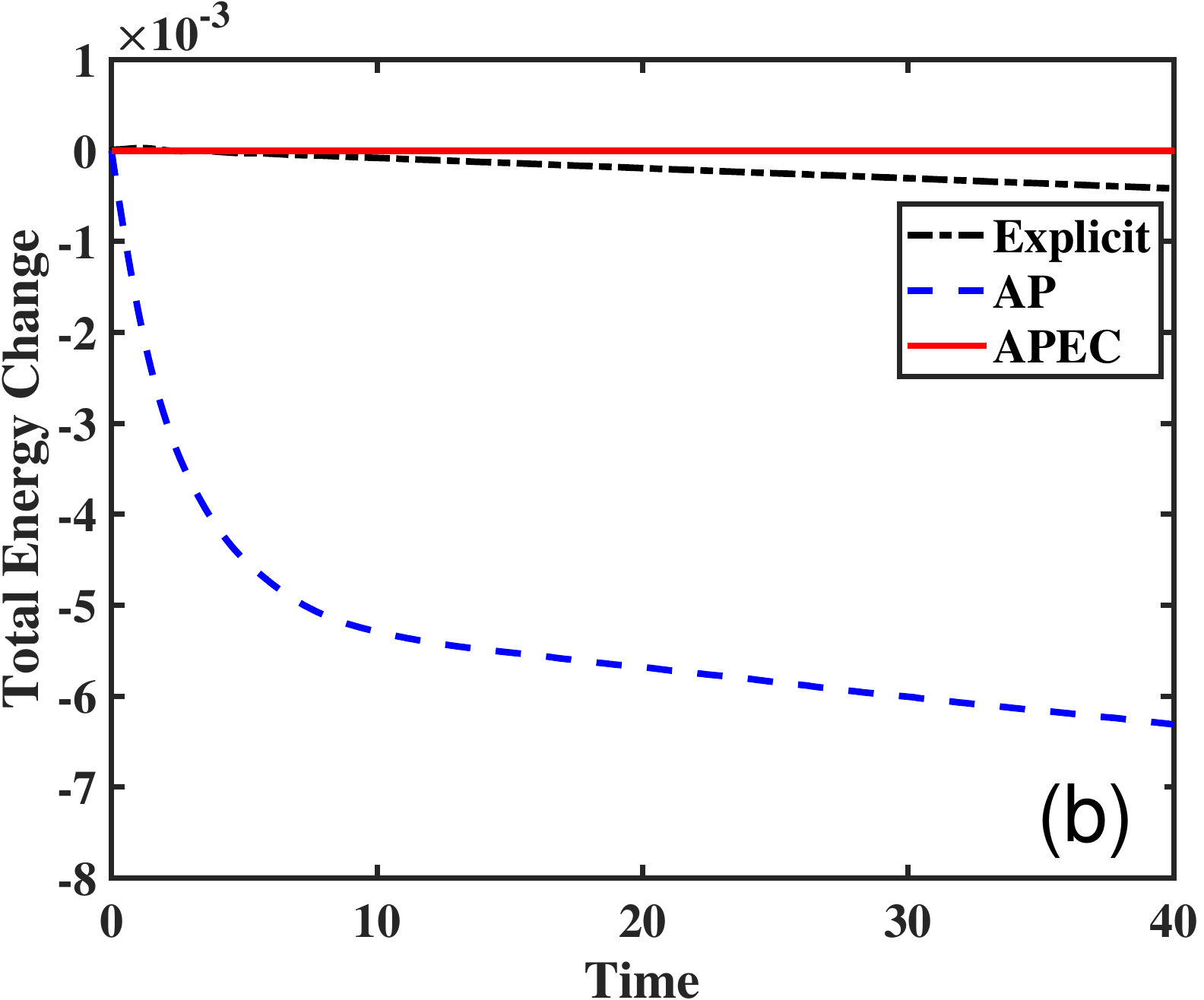}
\caption{Resolved case of the Landau damping calculated by the classical explicit, the AP and the APEC schemes. (a) Electric energy with time; (b) Total energy change with time.}
\label{fig:landau:p1}
\end{figure}

For practical applications, it is important to validate the effectiveness of the PIC methods for simulations with large time steps. Figure \ref{fig:landau:p2} presents the results of the under-resolved case for the three PIC schemes with the same $\lambda$, $N_x$ and $N$ but a large time step $\Delta t=2$. Apparently, both the AP and APEC algorithms are stable while the explicit algorithm is unstable when the time step is too large. The results of the AP-type schemes are consistent with those of energy-conserving schemes \cite{chen2011energy}.

\begin{figure}[H]
\centering
\includegraphics[width=0.45\linewidth]{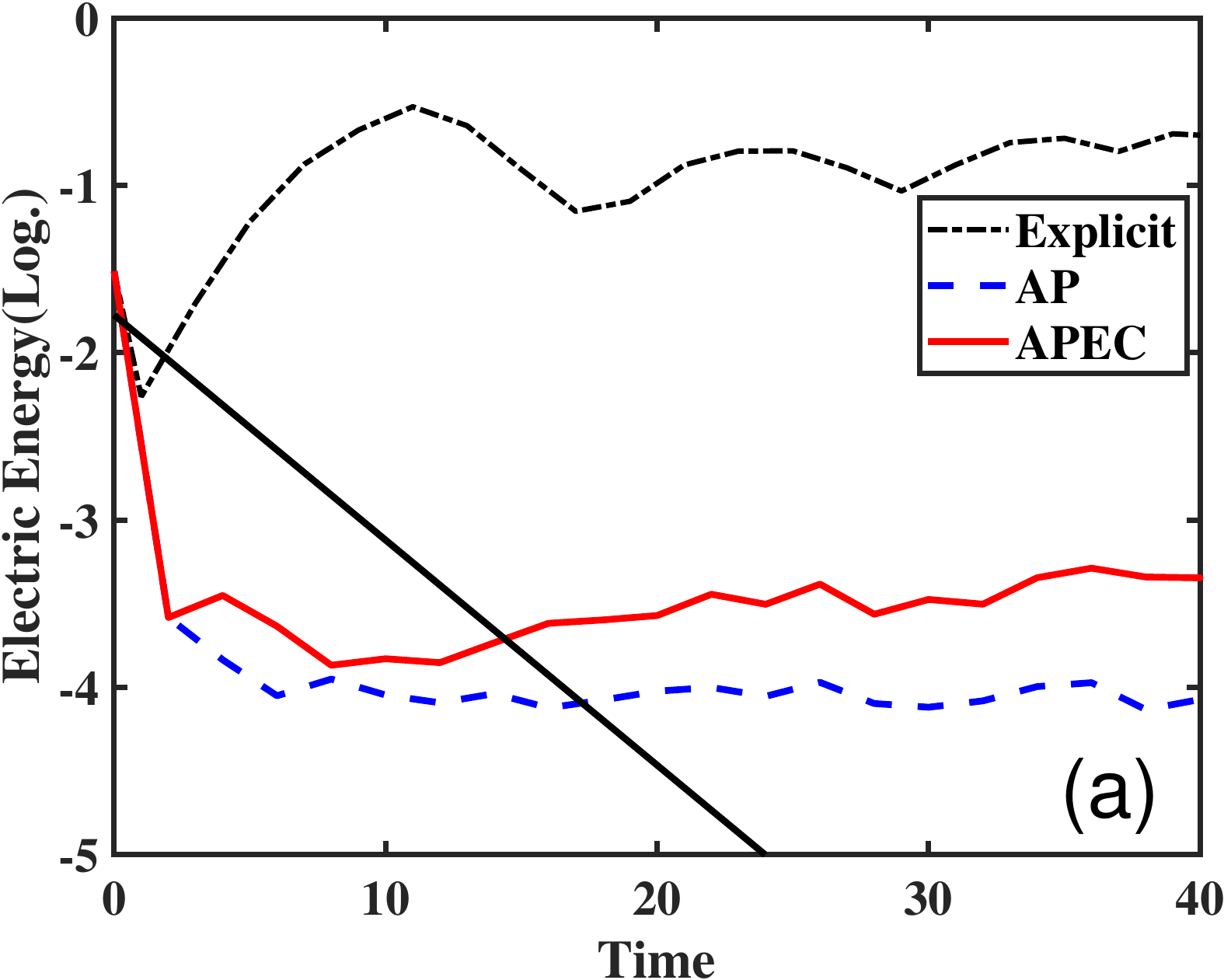}\hspace{2mm}
\includegraphics[width=0.455\linewidth]{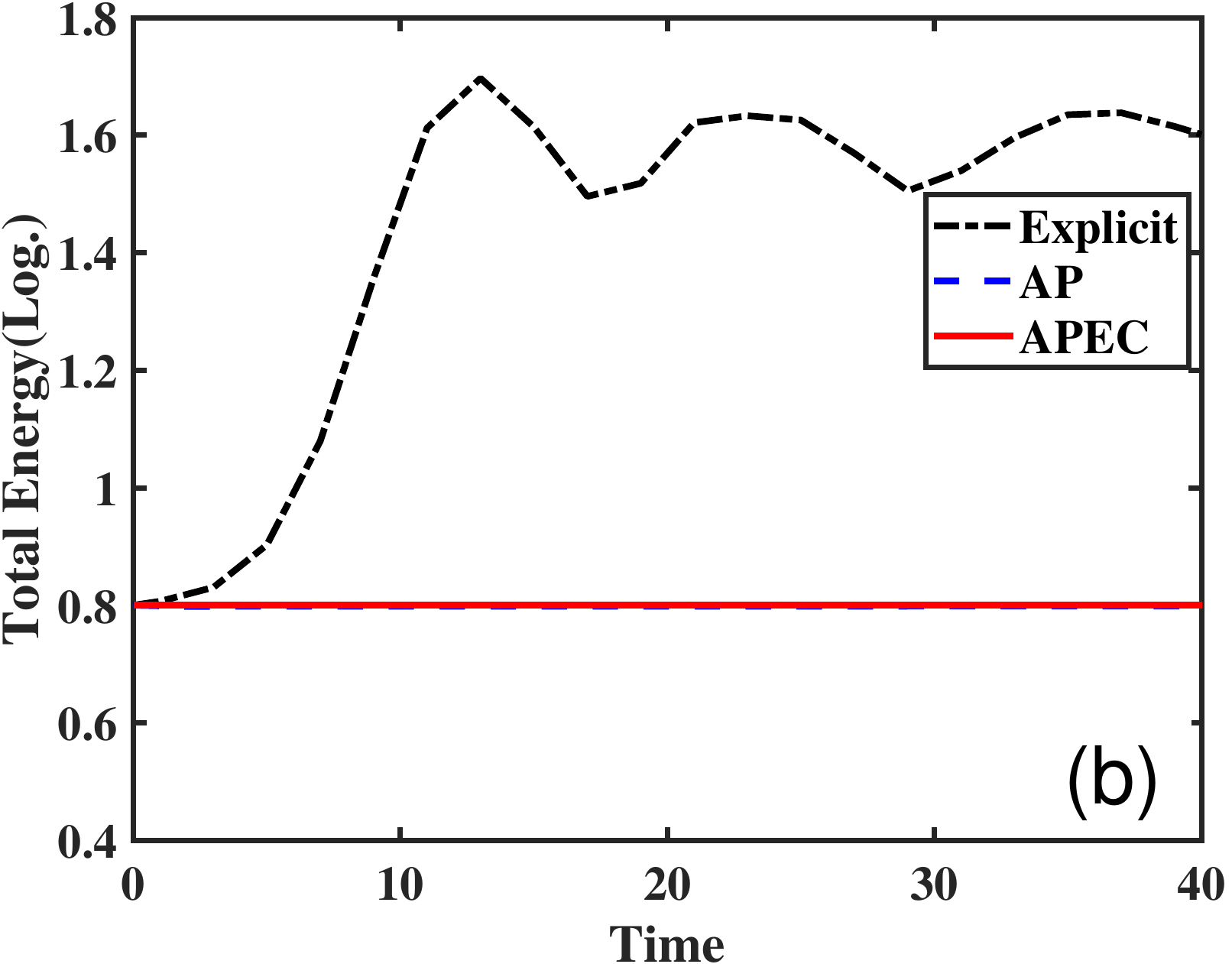}
\caption{Under-resolved case of the Landau damping calculated by the classical explicit, the AP and the APEC schemes. (a) Electric energy with time; (b) Total energy with time.}
\label{fig:landau:p2}
\end{figure}

\subsection{Bump-on-tail instability}

In the second benchmark test, we consider the bump-on-tail instability problem, where the velocity distribution has multiple peaks. The initial density of the Vlasov-Poisson equations is given by,
\begin{align*}
f_0(x,v) =g(v)\left[1+\delta \cos (\kappa x)\right],
\end{align*}
where function $g$ is the summation of two normal distributions with different means and variances,
\begin{align*}
\displaystyle
g(v) = Ce^{-v^2/2} +\alpha e^{- (v-v_d)^2/ (2v_t^2)},
\end{align*}
with constant $C$ being determined to satisfy $\int g(v) dv=1$.
Here we use the similar settings as in Ref. \cite{2010Asymptotic} with $\Omega_x= (0,20\pi)$, $\delta =0.04$, $\kappa =0.3$, $v_d=4.5$, $v_t=0.5$ and $\alpha=2/9$. The Poisson equation is endowed  with a homogeneous Dirichlet boundary condition.
For the resolved case, we set the Debye length $\lambda=1$, the plasm period $\tau_p=\lambda$, the total time $T=200$, the number of macro particles $N=1\times 10^5$ and the grid parameters $N_x=500$ and $\Delta t=0.01$. Figure \ref{fig:bump:p1} displays the electric energy and the total energy change with  time evolution. One can observe  similar evolutions of the electric energy for the three PIC algorithms, which are consistent with literature results \cite{2010Asymptotic,chen2011energy,
belaouar2009asymptotically}. Specifically, the instability increases rapidly from $t = 10$ and the electric energy attains the same maximum value at $t = 20$ for all numerical schemes.
The initial total energy of this system is $W_0= 92.0822$. Panel (b) of the figure shows that the total energy of the APEC method is conserved for the simulation time up to $T=200$, but the AP scheme dissipates the energy to $\sim 95\%$ of the initial energy at time $T$. The explicit scheme performs well in the energy conservation for this resolved case and the total energy has a small oscillation around $W_0$.

In Figure \ref{fig:bump:p1_2}, the velocity distributions for different times are displayed. Obviously, the results of the three PIC algorithms are nearly the same, and in agreement with those in literature \cite{belaouar2009asymptotically}. There are two peaks at earlier times and then the smaller one gradually disappeared  from $t=10$ when the instability increases.

\begin{figure}[H]
\centering
\includegraphics[width=0.45\linewidth]{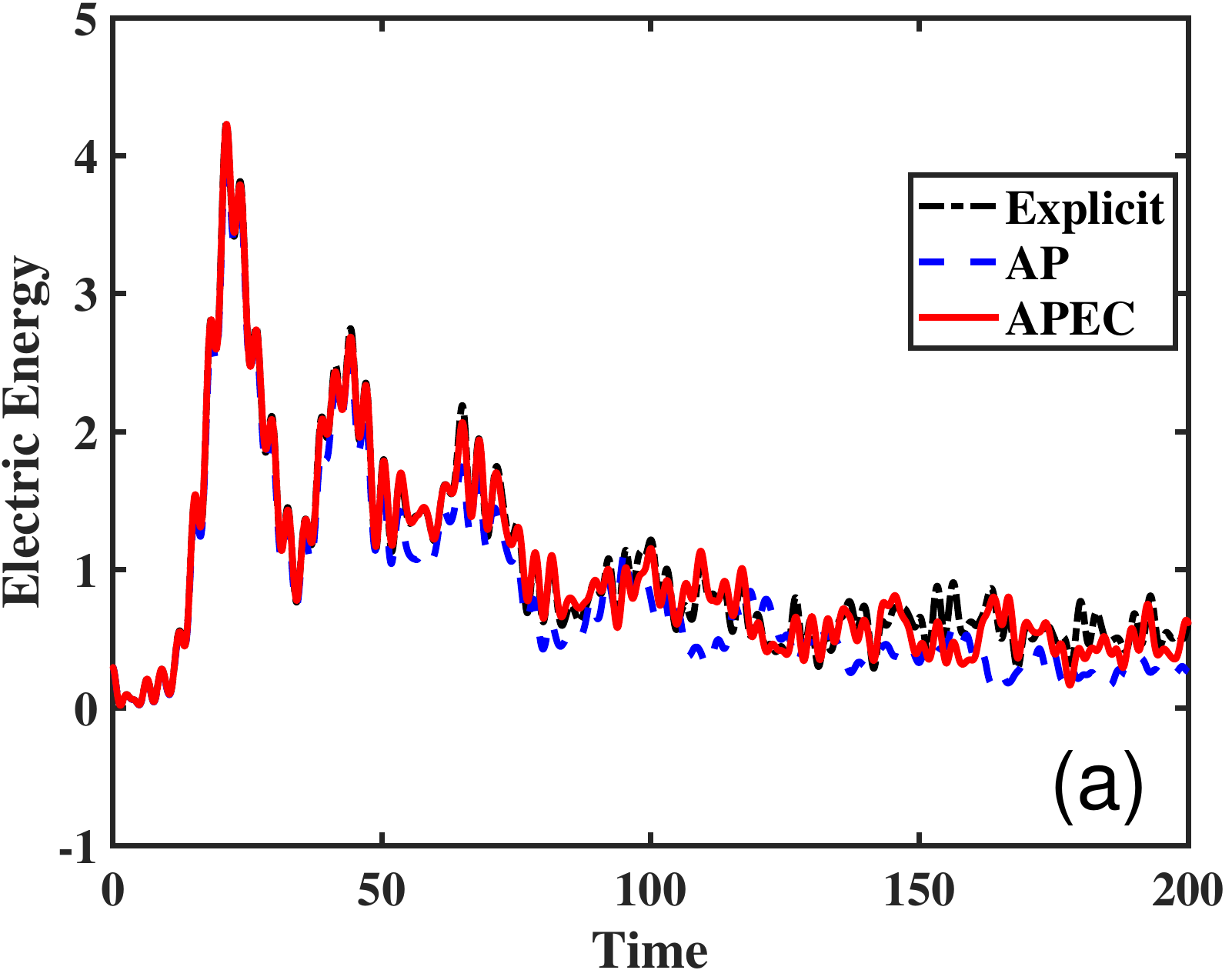}\hspace{2mm}
\includegraphics[width=0.45\linewidth]{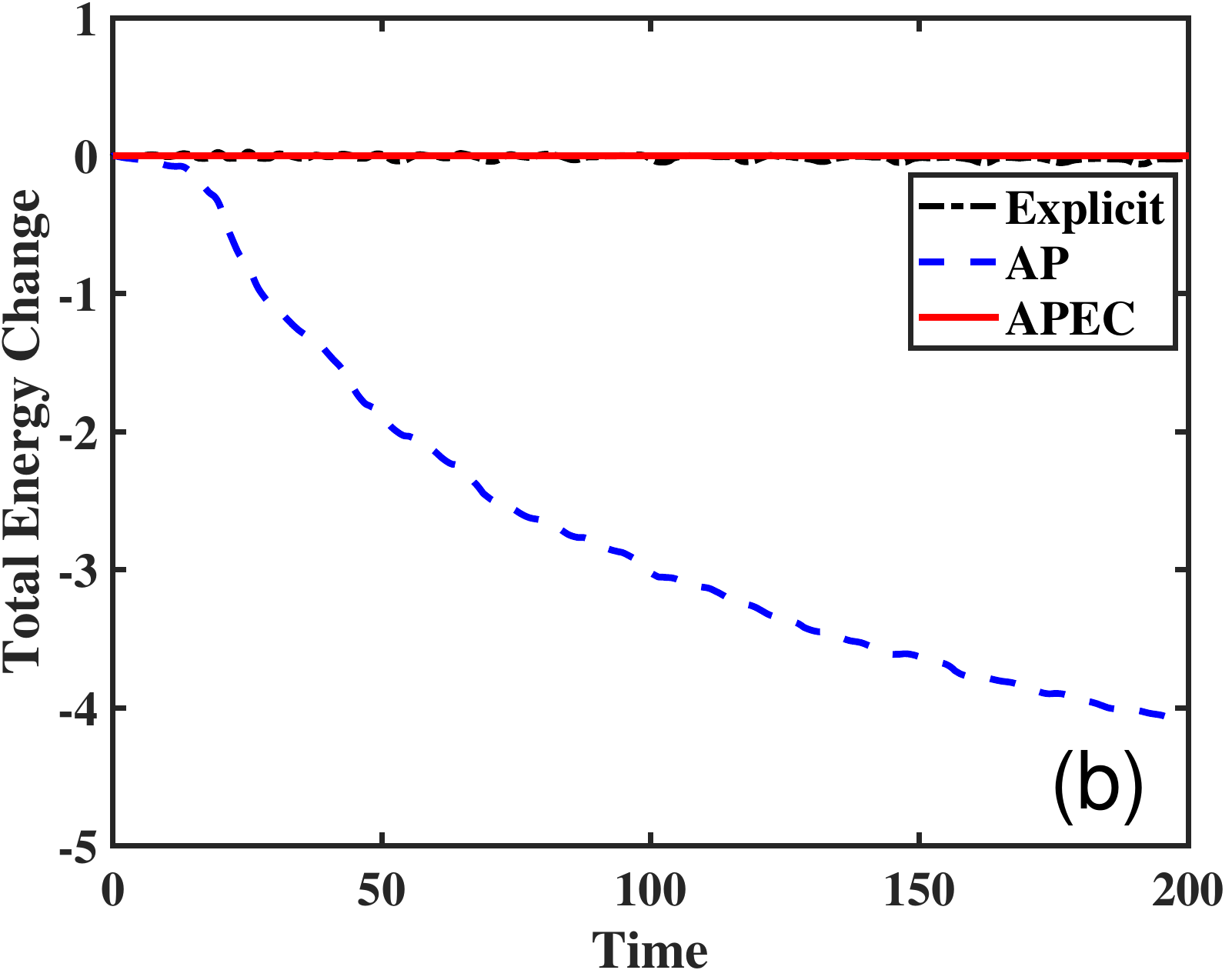}
\caption{Resolved case of the bump-on-tail instability calculated by the classical explicit, the AP and the APEC schemes. (a) Electric energy with time; (b) Total energy change with time.}
\label{fig:bump:p1}
\end{figure}

\begin{figure}[H]
\centering
\includegraphics[width=0.31\linewidth]{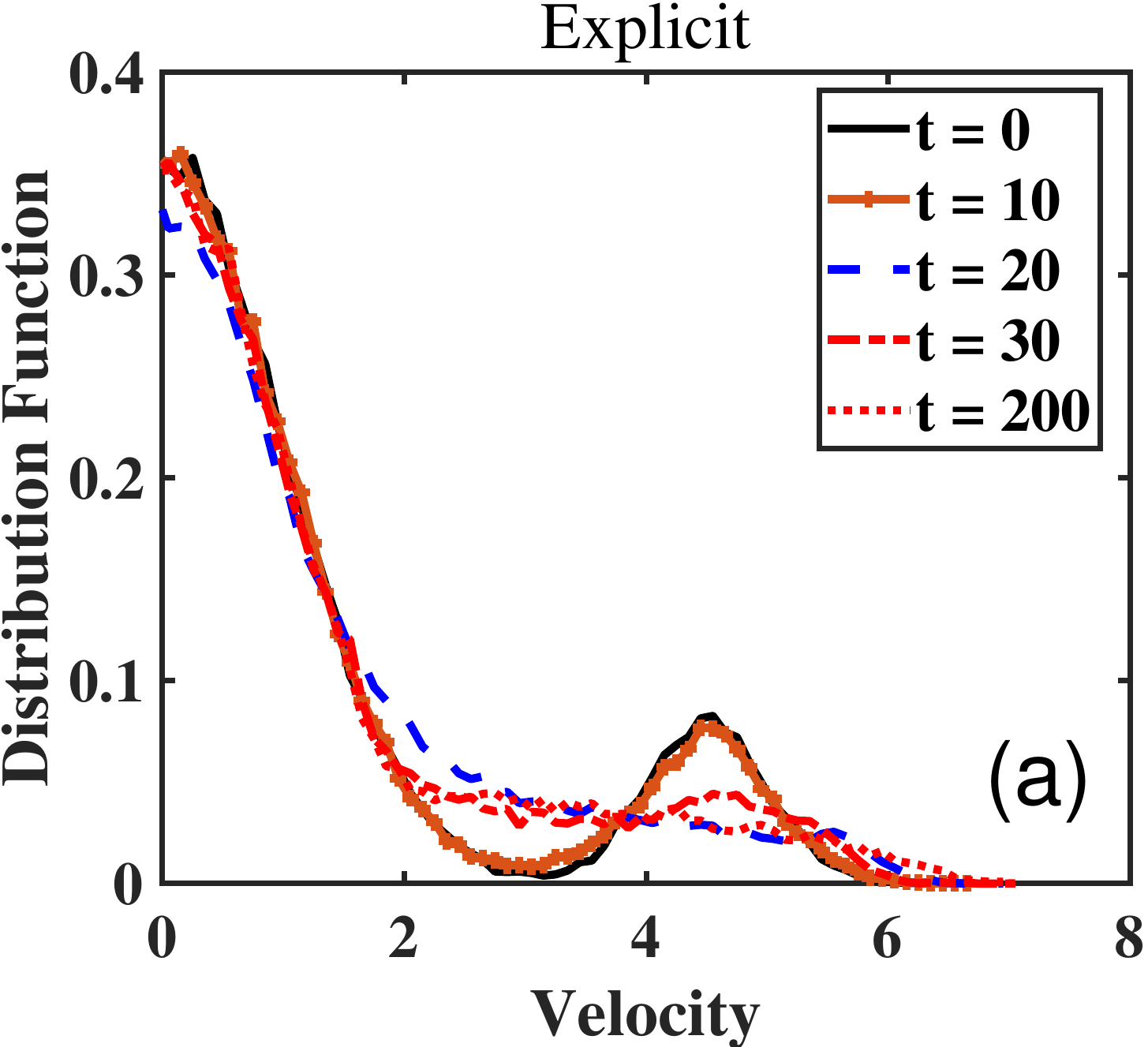}\hspace{1mm}
\includegraphics[width=0.31\linewidth]{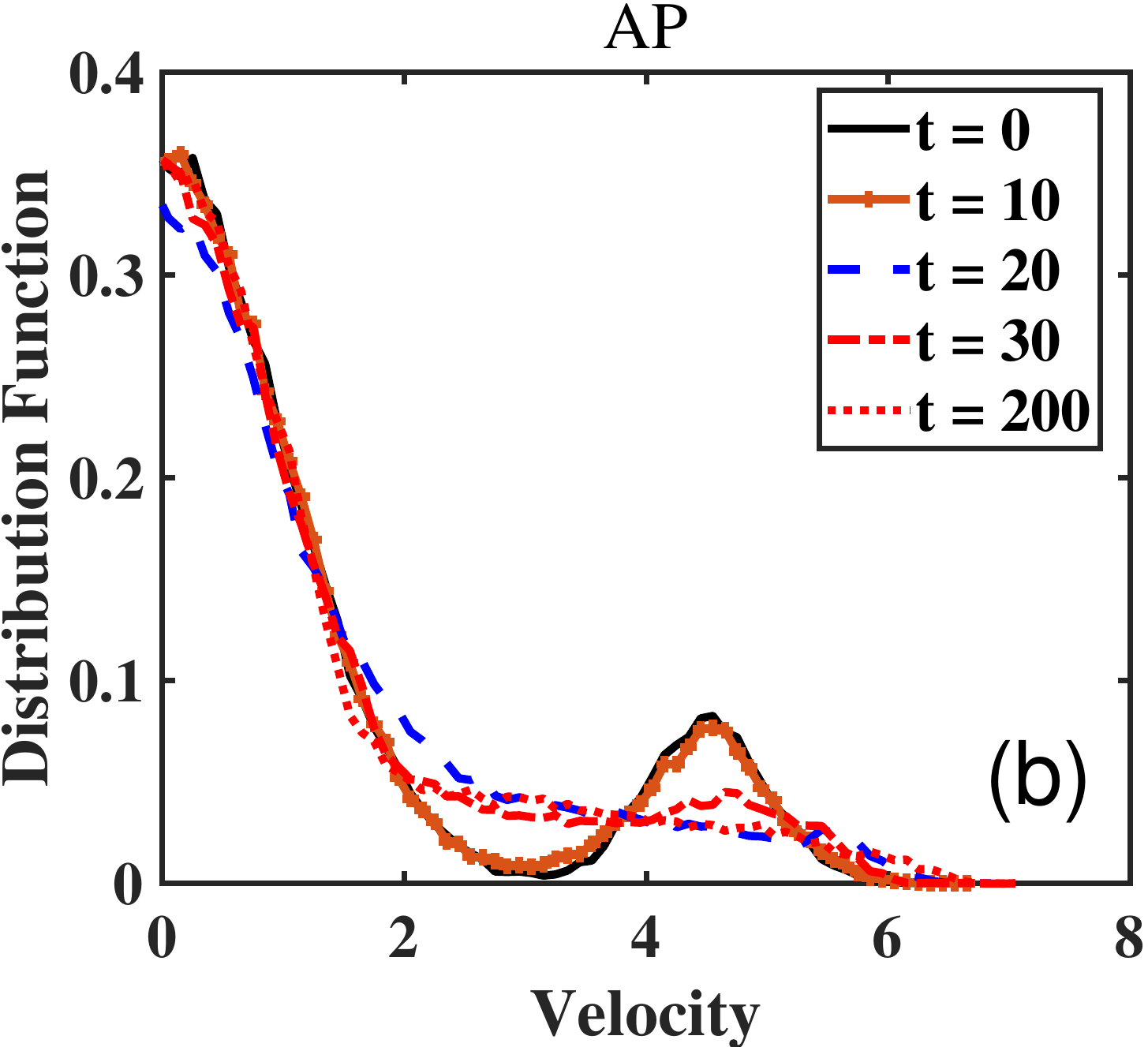}\hspace{1mm}
\includegraphics[width=0.31\linewidth]{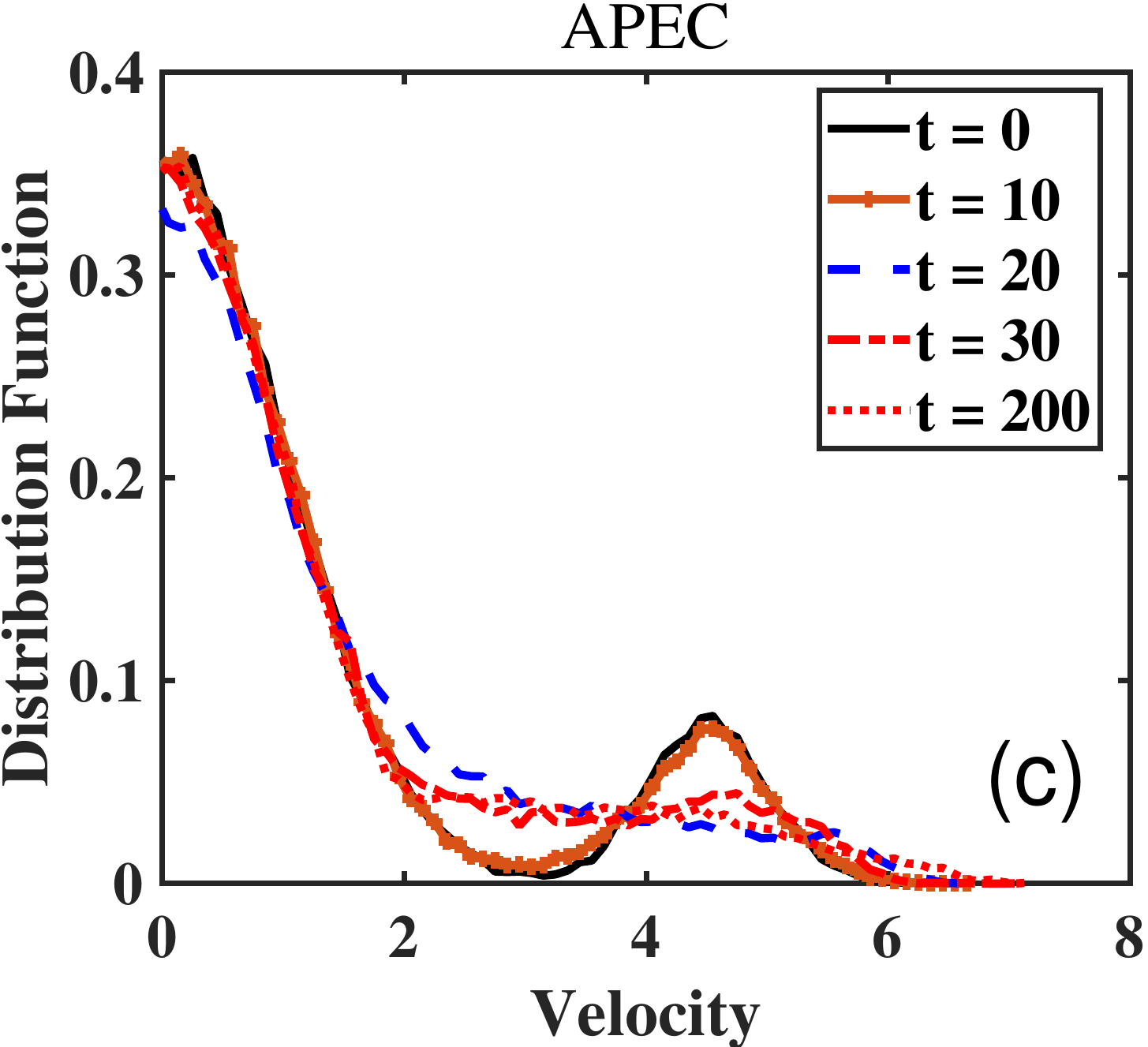}
\caption{ Velocity distributions at different times for the resolved case of the bump-on-tail instability. (a) the classical explicit algorithm;  (b) the AP algorithm; and (c) the APEC algorithm. }
\label{fig:bump:p1_2}
\end{figure}

The results of the under-resolved case are shown in Figure \ref{fig:bump:p2}, where we set the simulation parameters as the Debye length
$\lambda=0.1$ and the number of macro particles $N=1\times 10^6$. The grid parameters are $N_x=20$ and $\Delta t= 0.4$, which satisfy $\Delta x>\lambda$ and $\Delta t>\tau_p$. In this setup, the initial total energy is $W_0=106.1615$. As expected, the classical explicit algorithm is unstable and blows up in this case due to the large time step, while both AP and APEC algorithms provide the stable results. Again, the total energy of the APEC algorithm is conserved for the whole simulation time for this under-resolved case. With the same setup,
the velocity distributions at different times are displayed in Figure \ref{fig:bump:p2_2}. One can see that the explicit algorithm heats up the system quickly
such that the macro particles with smaller velocities vanish soon.
One can see that the AP and APEC algorithms give  similar results, in consistent with those in Refs. \cite{belaouar2009asymptotically,2010Asymptotic}.
Since the APEC scheme has no energy dissipation, the number of the macro particles with large velocities ($v_p>8$) is slightly bigger that the results of the AP scheme.
This is in agreement with our intuition.

Additionally, in order to validate the APEC method with different time steps, we show in Figure \ref{fig:bump:p3} the results of parameters $\lambda=0.1$ and $N_x=20$ for $\Delta t=0.01, 0.1$ and $0.2$, where the electric energy with the time evolution is calculated. One can clearly observe that the APEC algorithm is stable and the electric energy curve has similar behavior for different time steps.
\begin{figure}[H]
\centering
\includegraphics[width=0.45\linewidth]{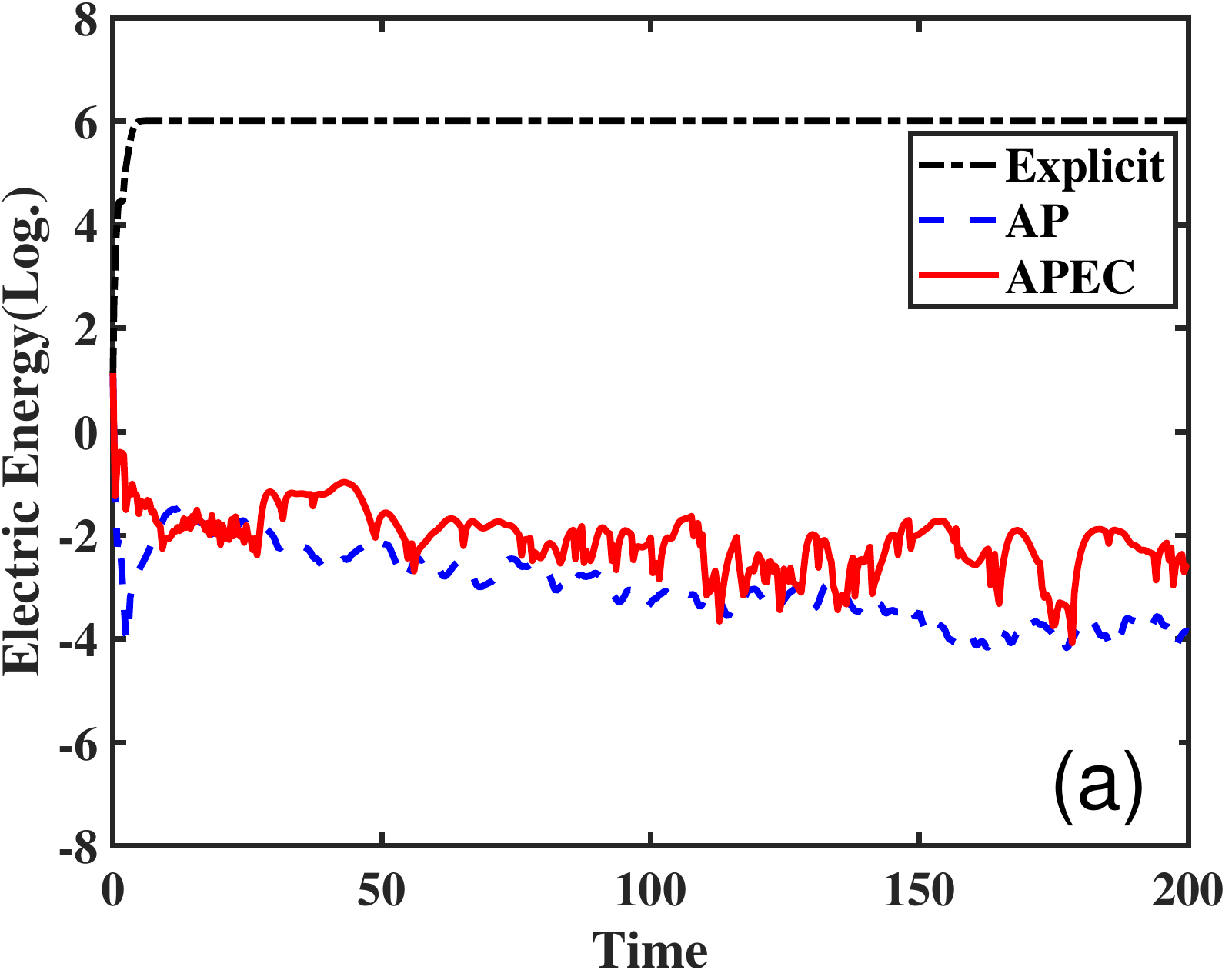}\hspace{2mm}
\includegraphics[width=0.445\linewidth]{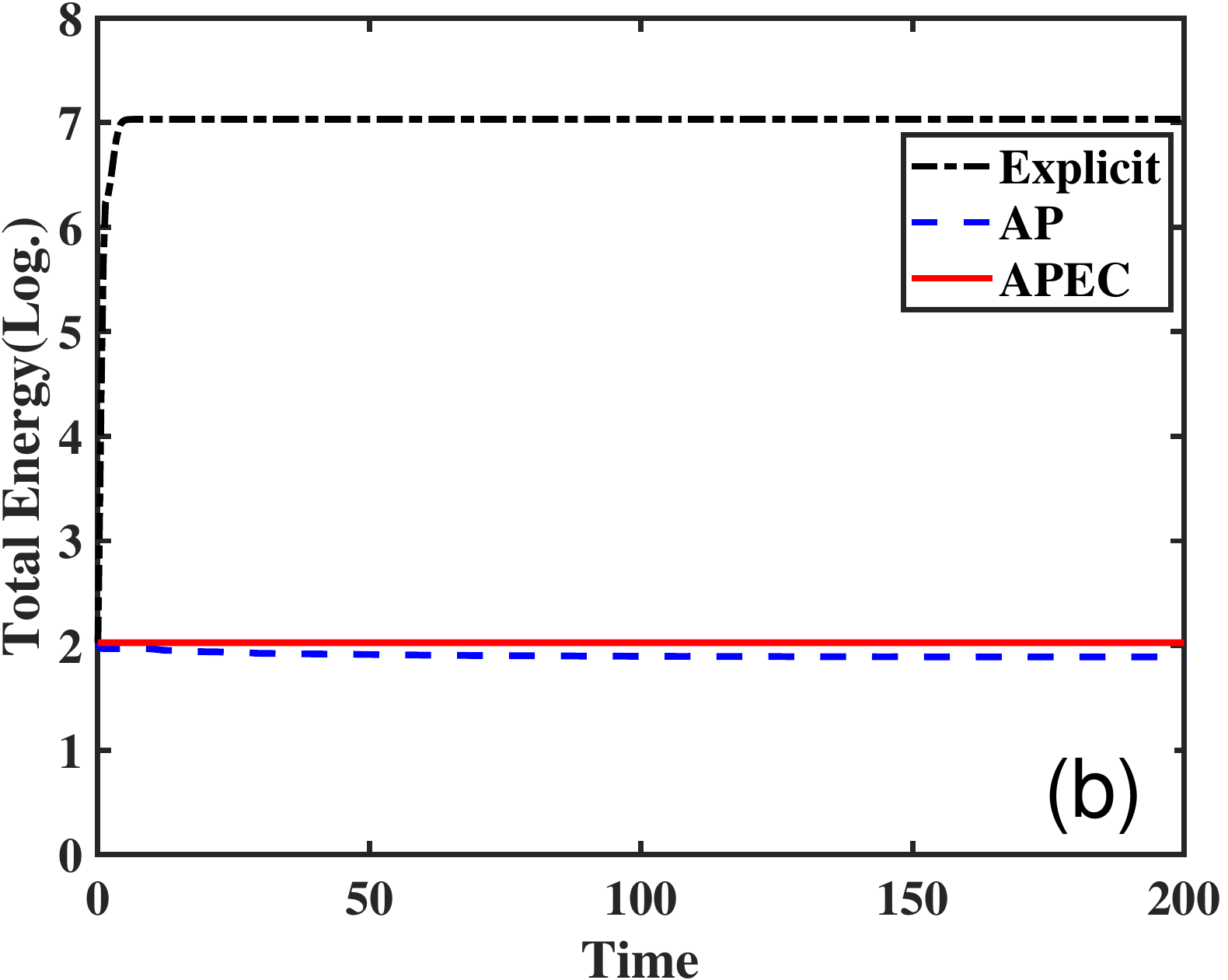}
\caption{Under-resolved case of the bump-on-tail instability calculated by the classical explicit, the AP and the APEC schemes. (a) Electric energy with time; (b) Total energy with time.}
\label{fig:bump:p2}
\end{figure}

\begin{figure}[H]
\centering
\includegraphics[width=0.31\linewidth]{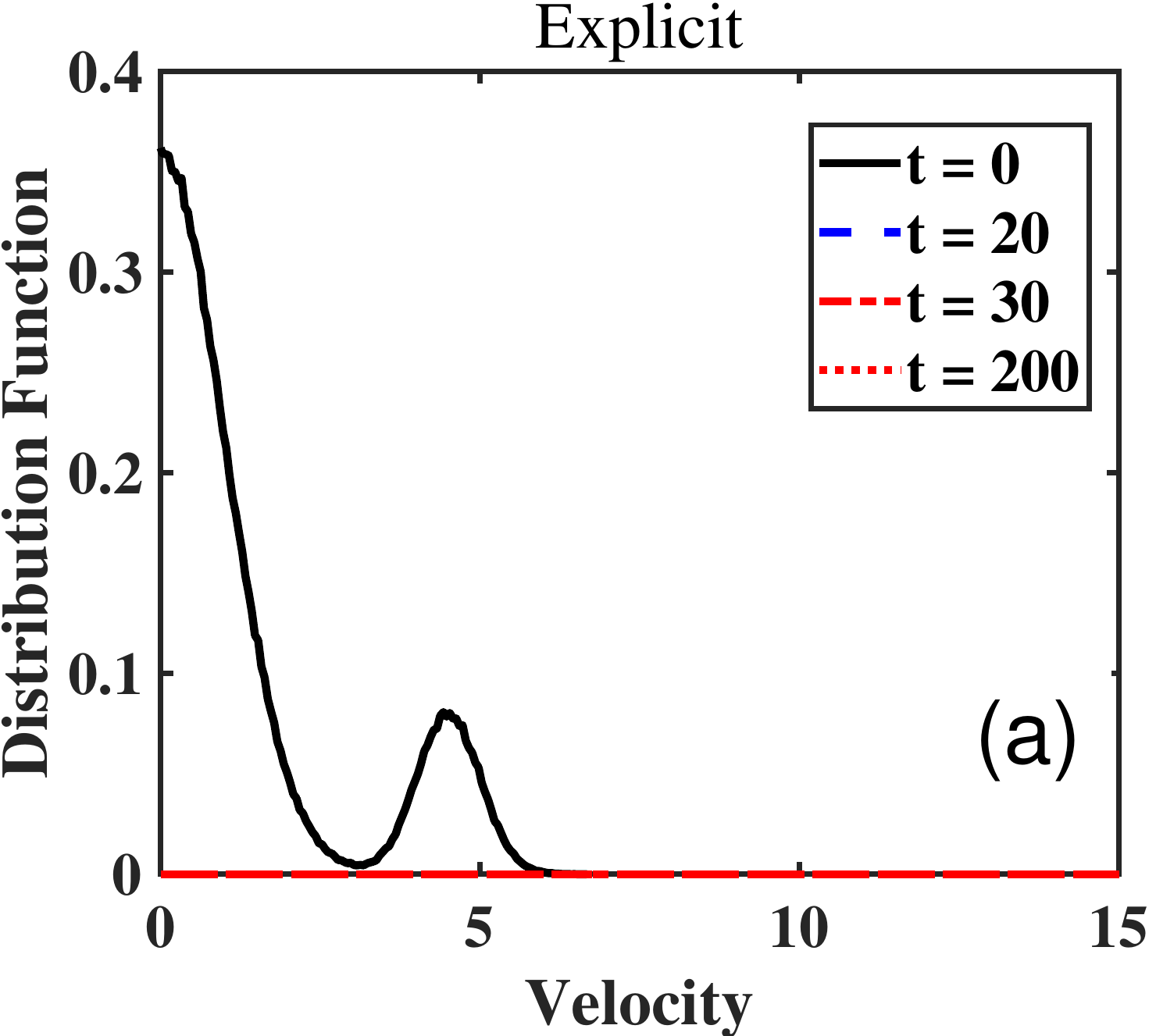}
\includegraphics[width=0.31\linewidth]{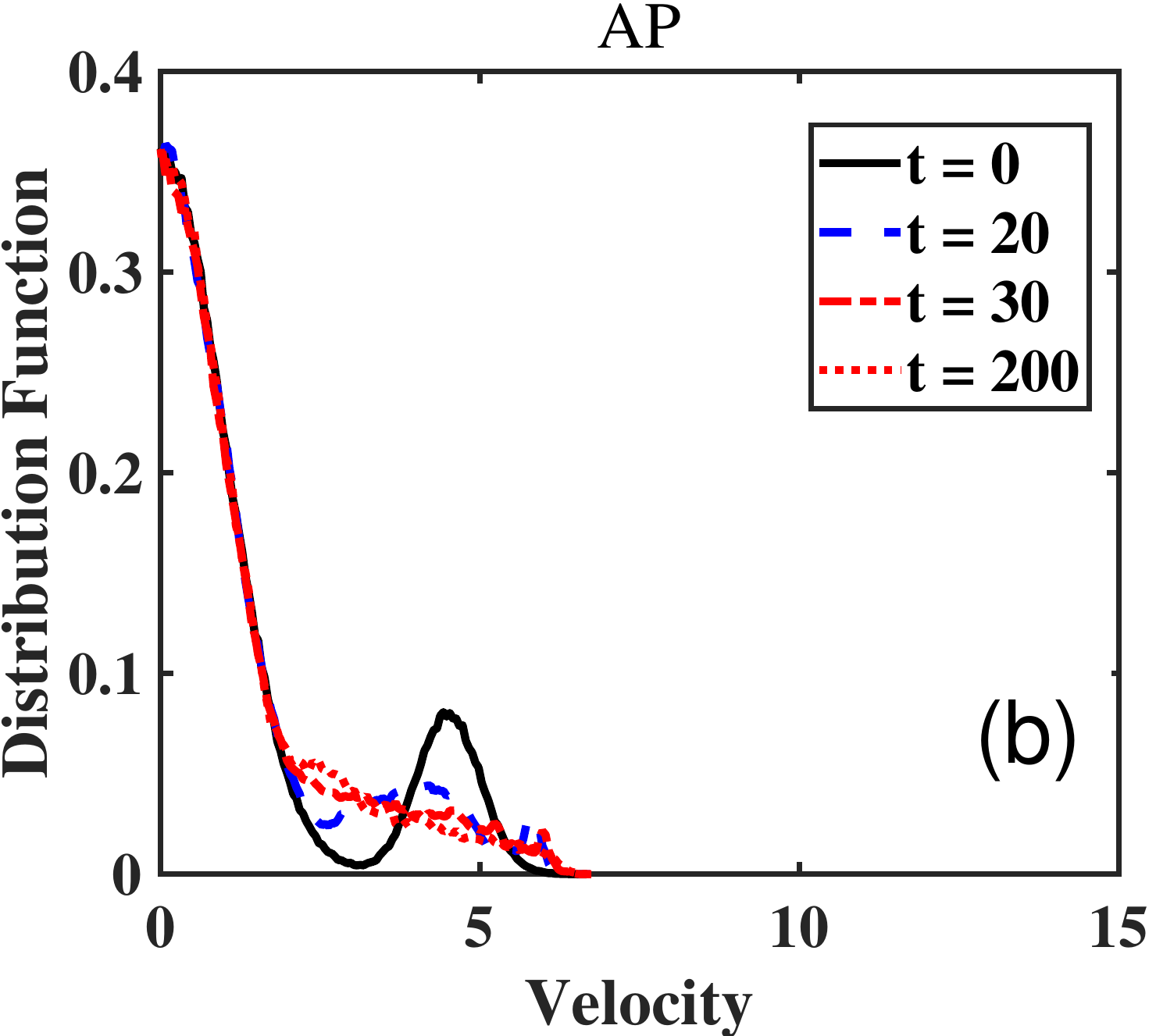}
\includegraphics[width=0.31\linewidth]{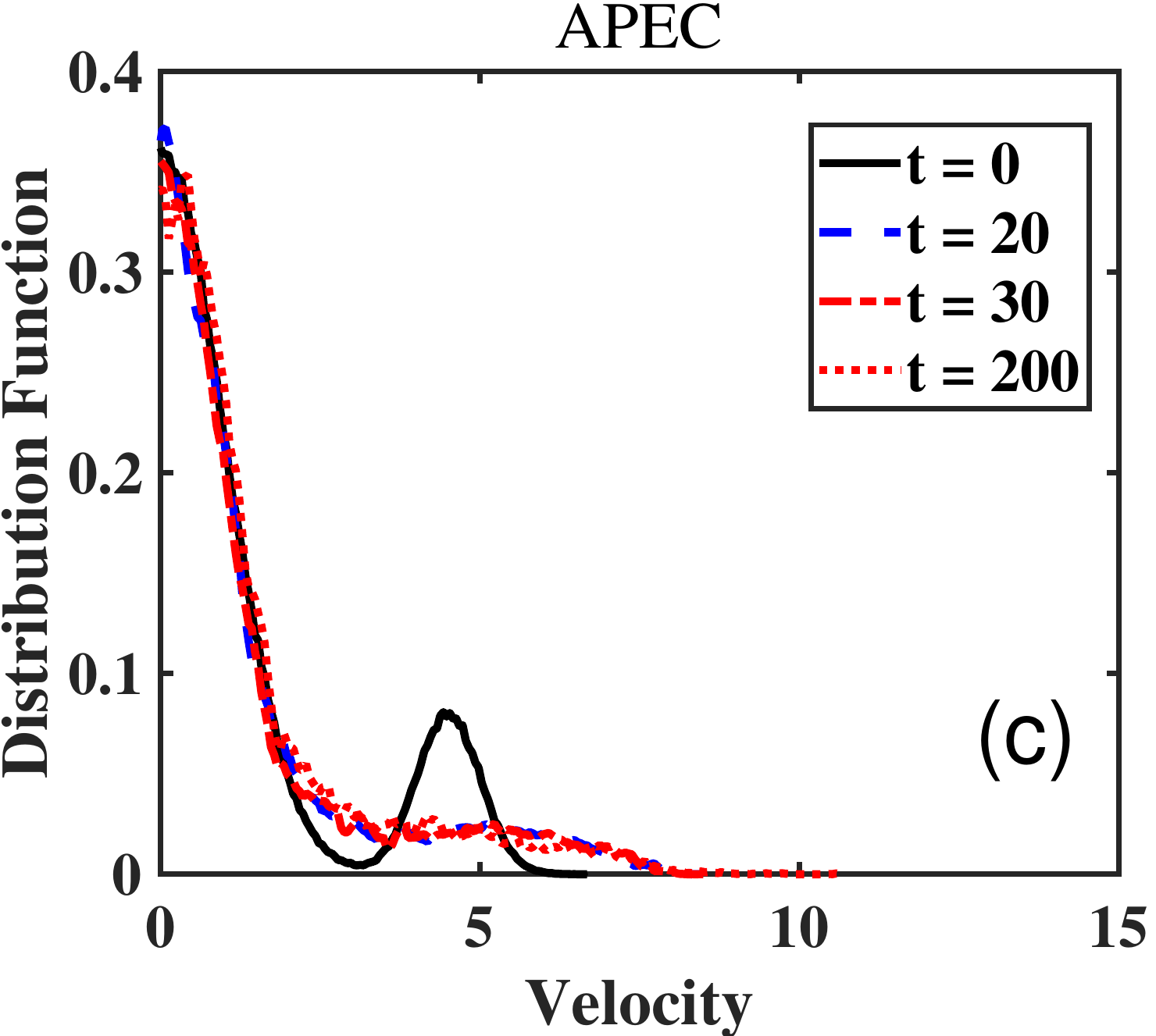}
\caption{Velocity distributions at different times for the under-resolved case of the bump-on-tail instability. (a) the classical explicit algorithm;  (b) the AP algorithm; and (c) the APEC algorithm.}
\label{fig:bump:p2_2}
\end{figure}

\begin{figure}[H]
\centering
\includegraphics[width=0.45\linewidth]{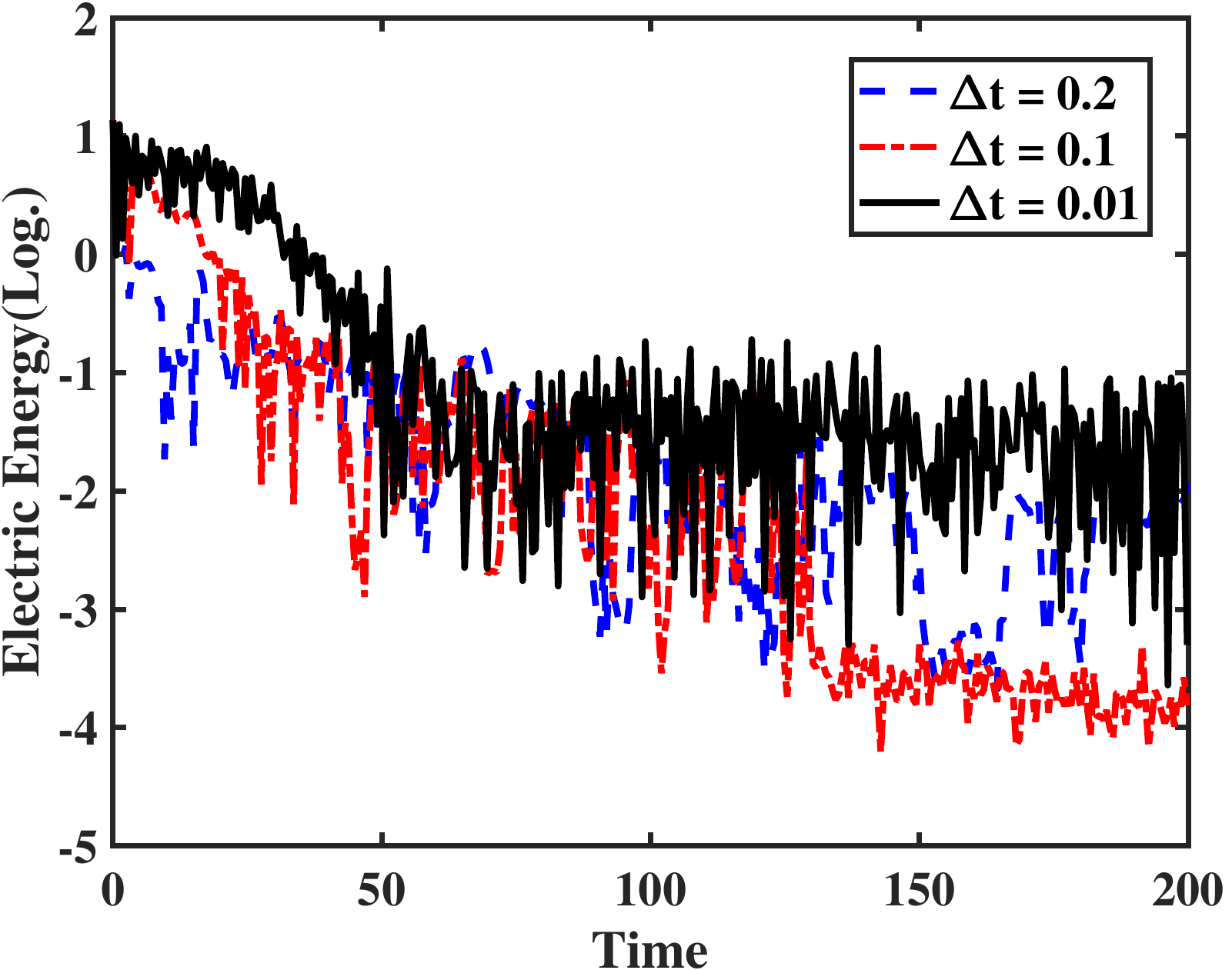}
\caption{Electric energy with time by the APEC algorithm for the bump-on-tail instability with different time steps $\Delta t=0.01, 0.1$ and $0.2$.  }
\label{fig:bump:p3}
\end{figure}

\subsection{Two-stream instability}
The third benchmark test is the two-stream instability which has been often studied in plasma literature \cite{ho2018physics,chen2011energy,stix1992waves}.
The initial density distribution is
\begin{align*}
f_0(x,v) =\frac{1}{2\sqrt{2\pi}\sigma}\left[e^{-(v+v_b)^2/2\sigma^2} + e^{-(v-v_b)^2/2\sigma^2}\right] \left( 1+\alpha \cos x\right),
\end{align*}
with $v_b=\sqrt{3}/2$ being the beam speed, $x\in (0,2\pi)$ and $\sigma=0.008$. For the resolved case, we set $\lambda =0.5$, $\alpha=0.005$, $N_x=64$, $N=1\times 10^5$ and $\Delta t=0.02$. The initial total energy is $2.3565$. The periodic boundary condition is used for the Poisson equation. The results are displayed in Figure \ref{fig:twostream:p3}. Again, all the methods give similar results and have similar growing rates on the electric energy at the begin. However, when $t$ is about 9, the total energies change significantly for the results predicted by the explicit and the AP algorithms due to the interaction of plasmas of opposite speeds. During the interaction and the following run, the APEC method conserves  the energy and shows the best performance among the algorithms. Without the energy conservation, the AP algorithm has dissipation. Interestly, the explicit method dissipates the energy during the interaction, but returns to the initial energy at later time. Figure \ref{fig:twostream:p6} presents the distributions of the macro particles in the phase space for three snapshots. Clearly, during the growth period of the instability with the same rate shown in Figure \ref{fig:twostream:p3}(a),  the three PIC algorithms predict similar results. But after the instability getting saturated at $t\approx10$, the particle distributions of the three methods become pretty different in some portions of the phase space.

\begin{figure}[H]
\centering
\includegraphics[width=0.45\linewidth]{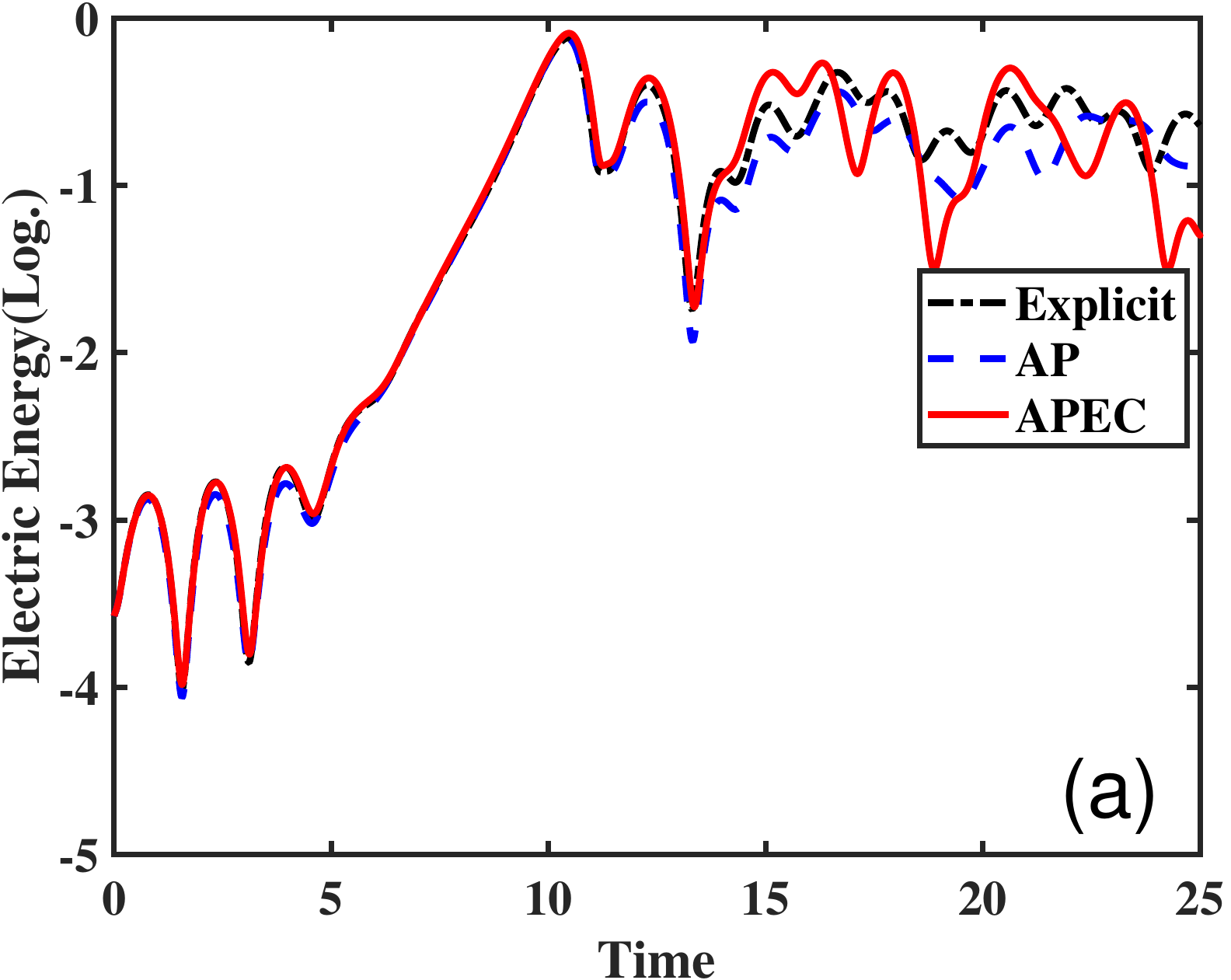}\hspace{2mm}
\includegraphics[width=0.465\linewidth]{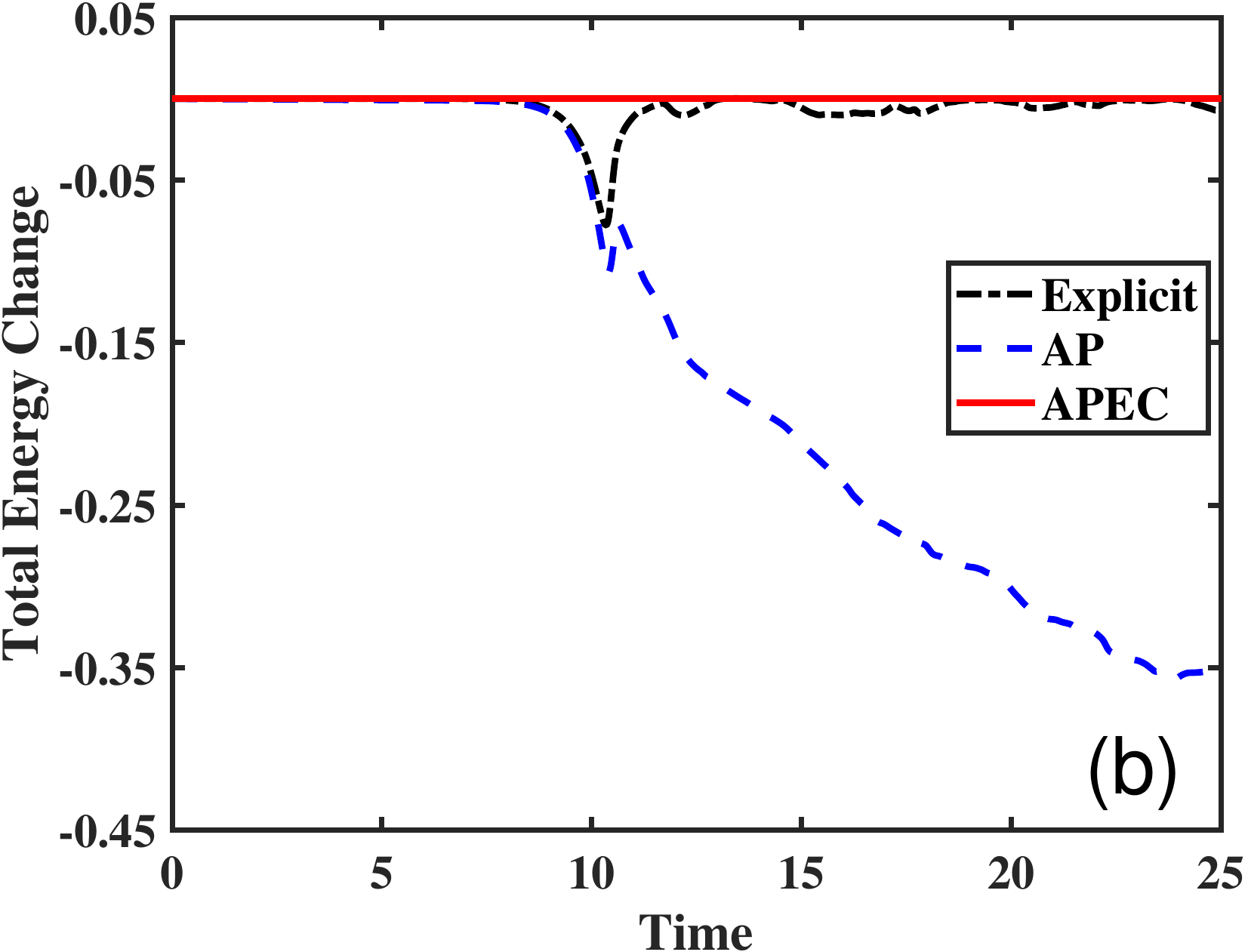}
\caption{Resolved case of the two-stream instability calculated by the classical explicit, the AP and the APEC schemes. (a) Electric energy with time; (b) Total energy with time.}
\label{fig:twostream:p3}
\end{figure}

\begin{figure}[H]
\centering
\includegraphics[width=0.32\linewidth]{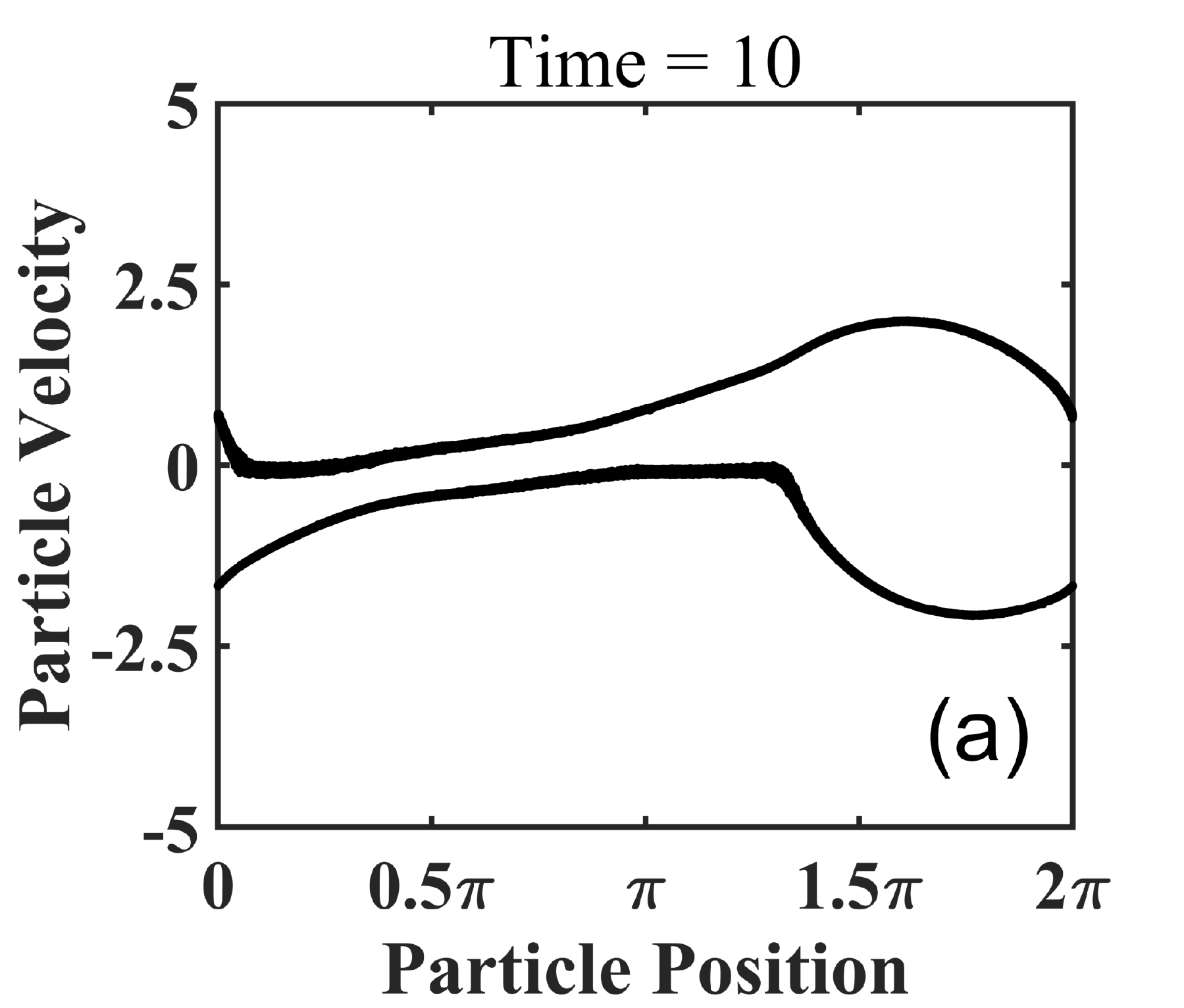}
\includegraphics[width=0.32\linewidth]{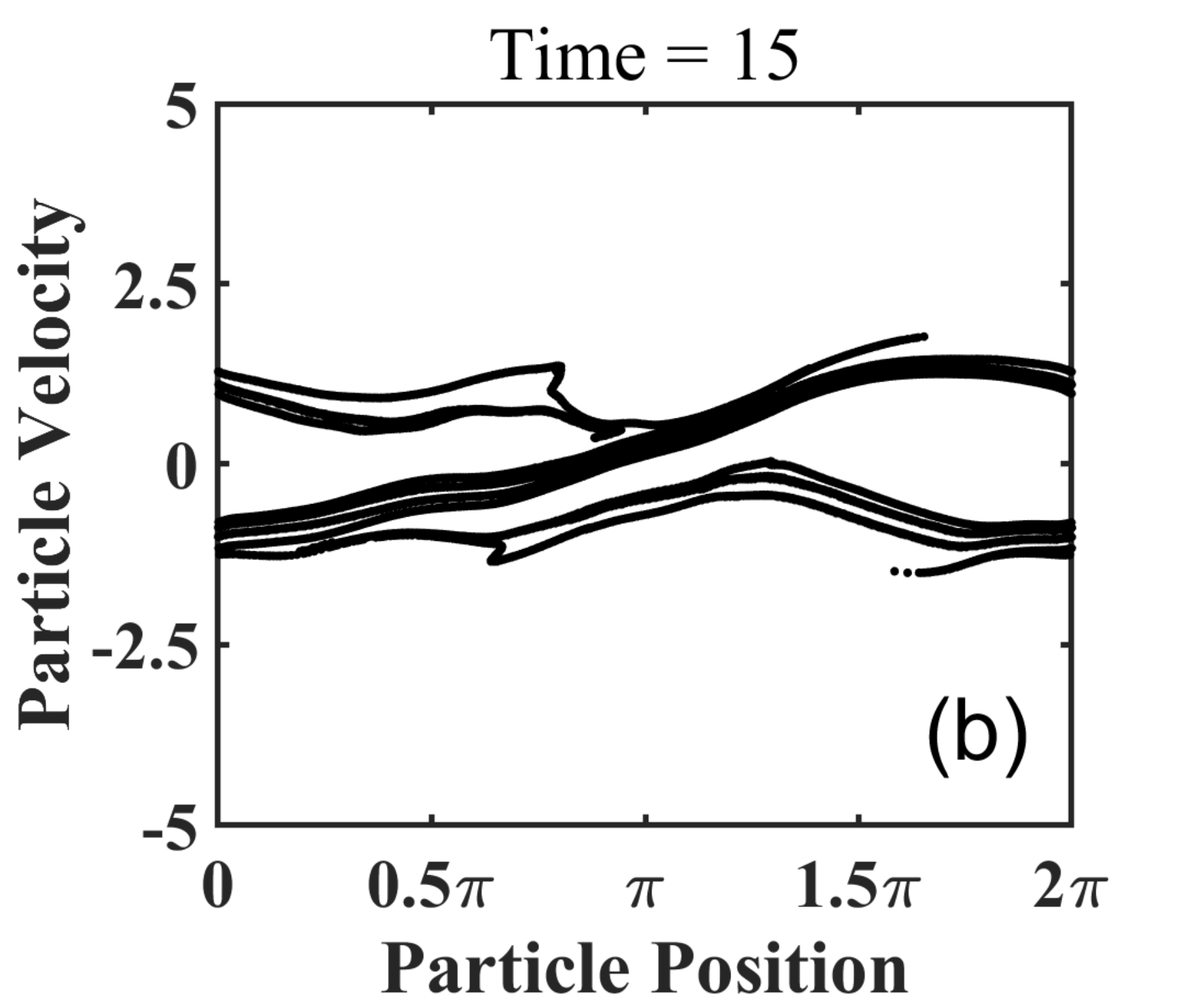}
\includegraphics[width=0.32\linewidth]{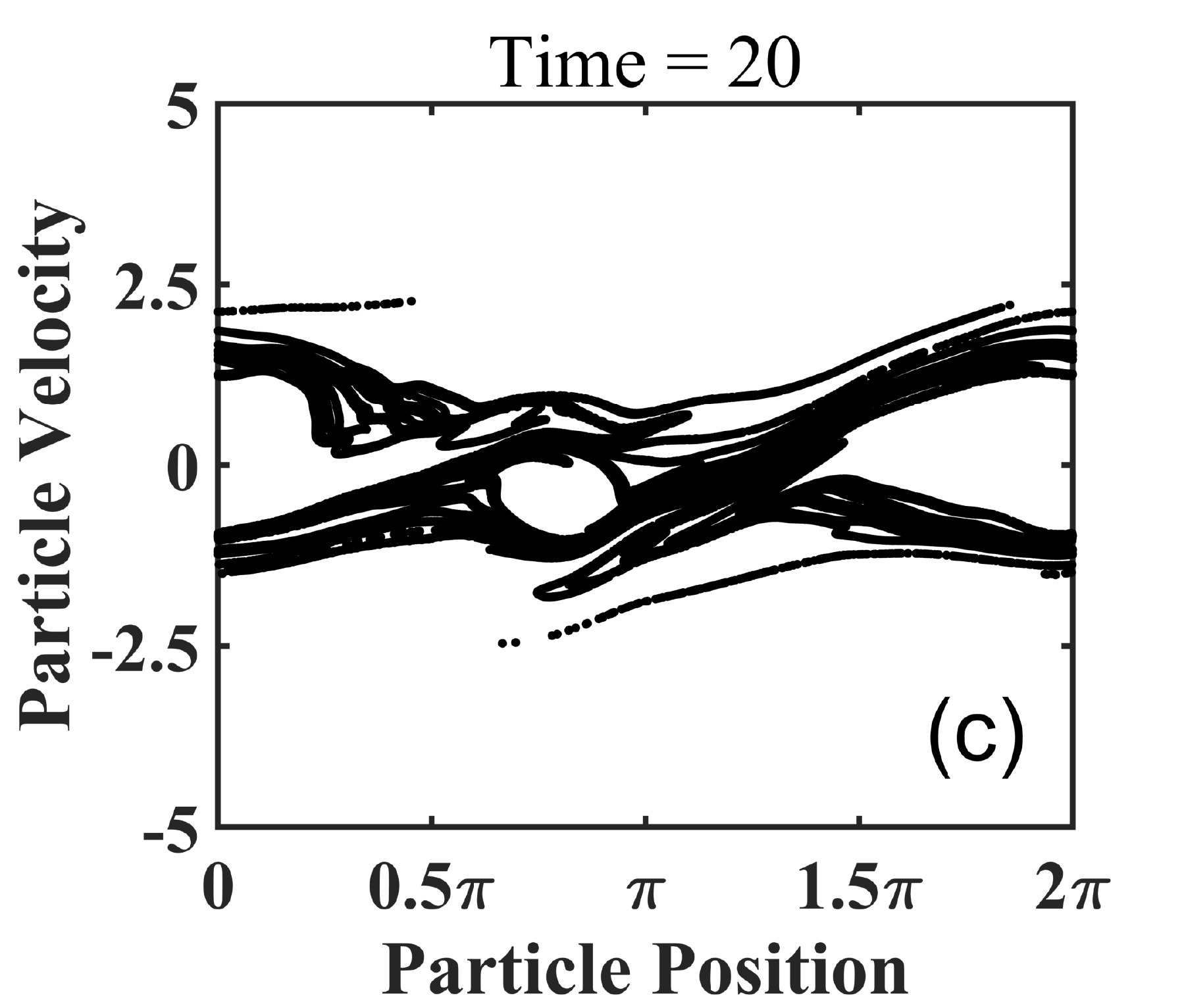}\\
\vspace{0.15mm}
\includegraphics[width=0.32\linewidth]{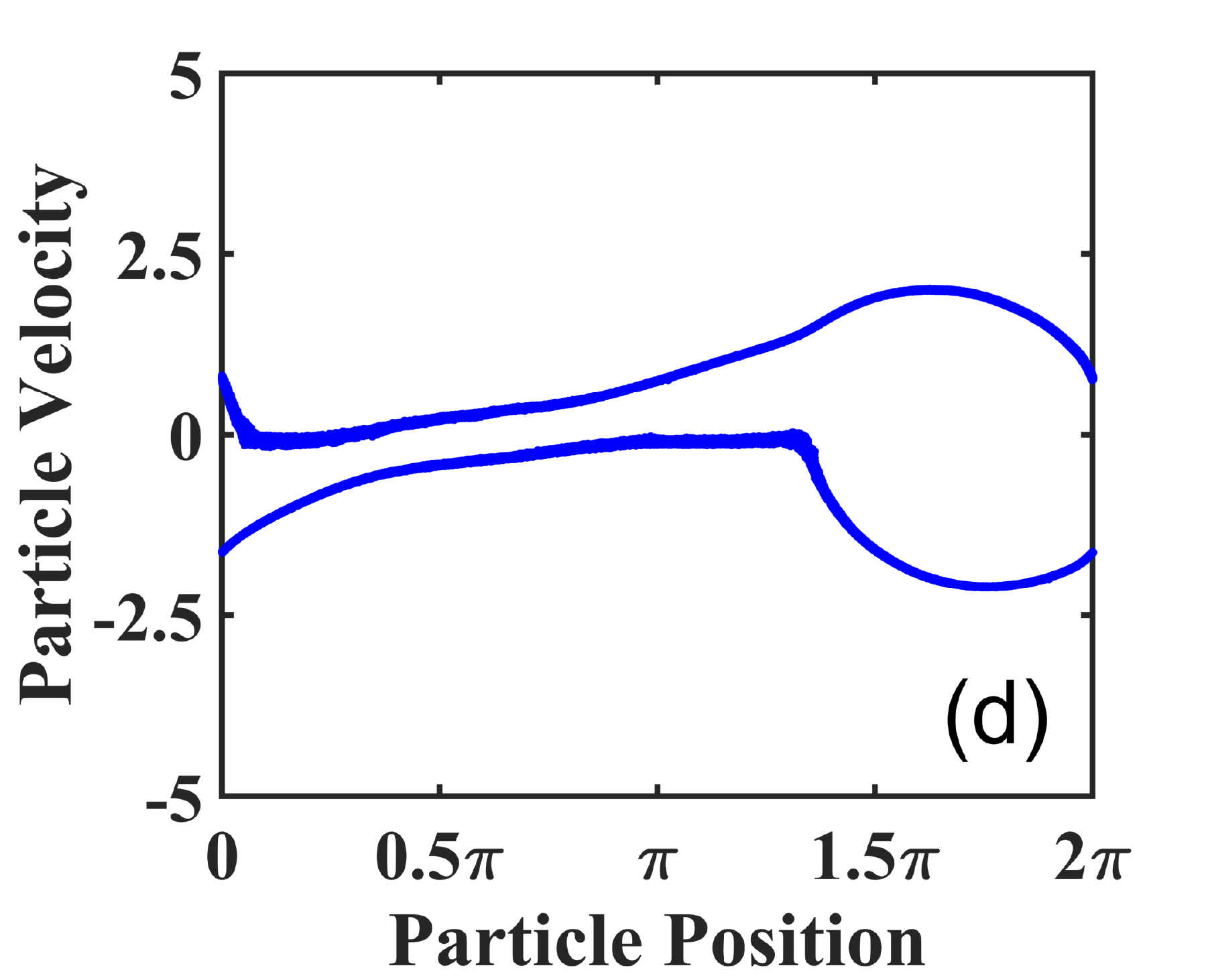}
\includegraphics[width=0.32\linewidth]{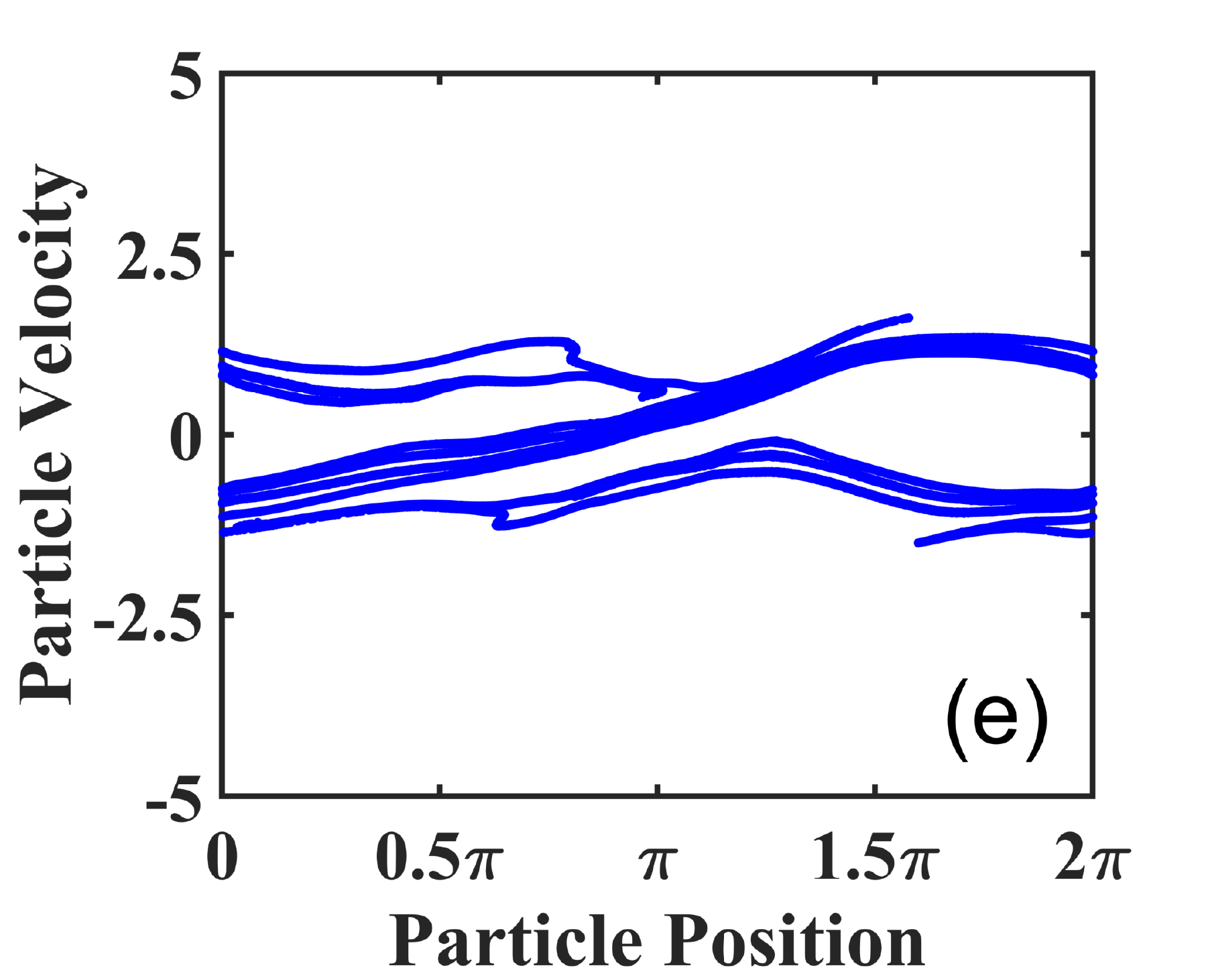}
\includegraphics[width=0.32\linewidth]{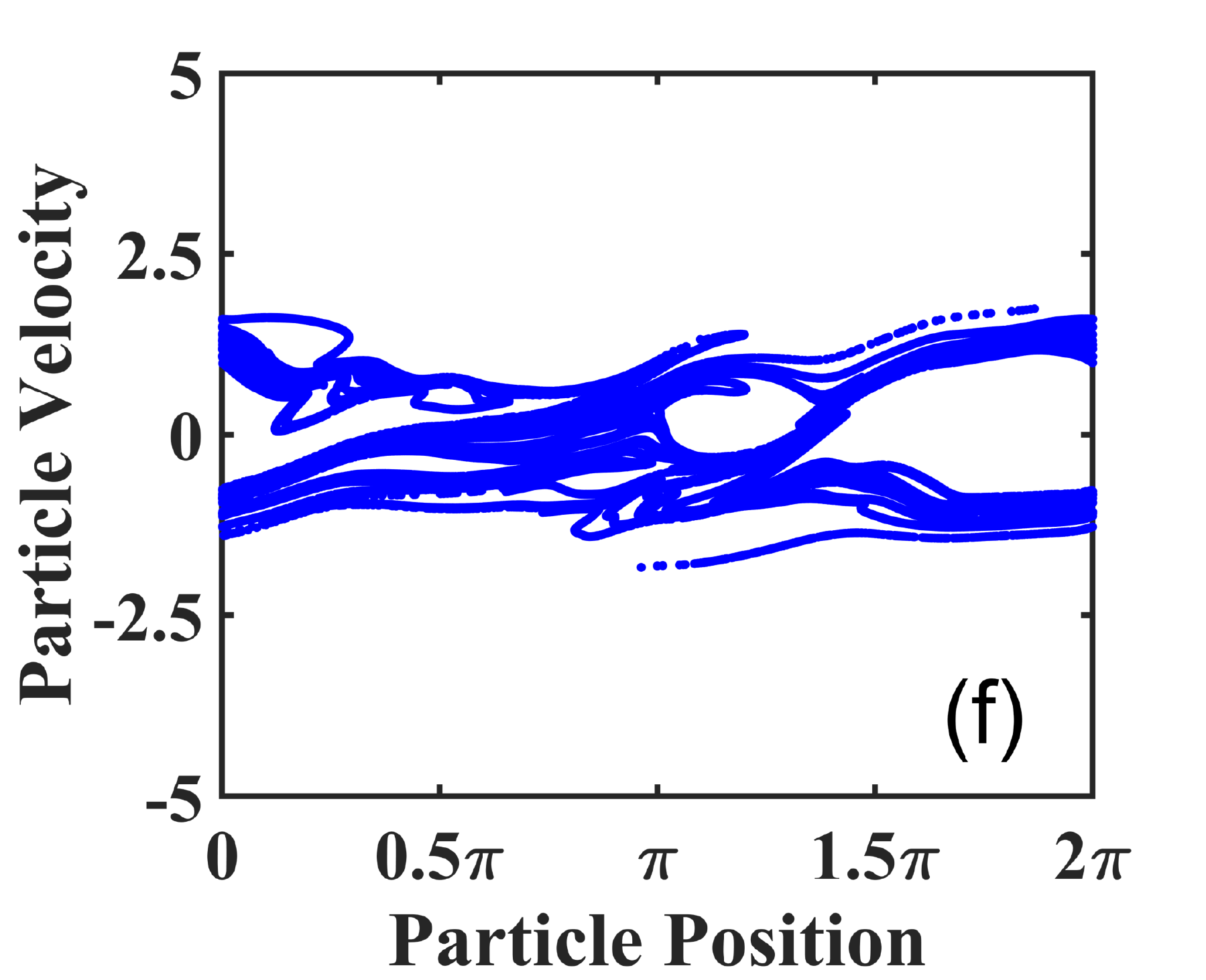}\\
\vspace{0.15mm}
\includegraphics[width=0.32\linewidth]{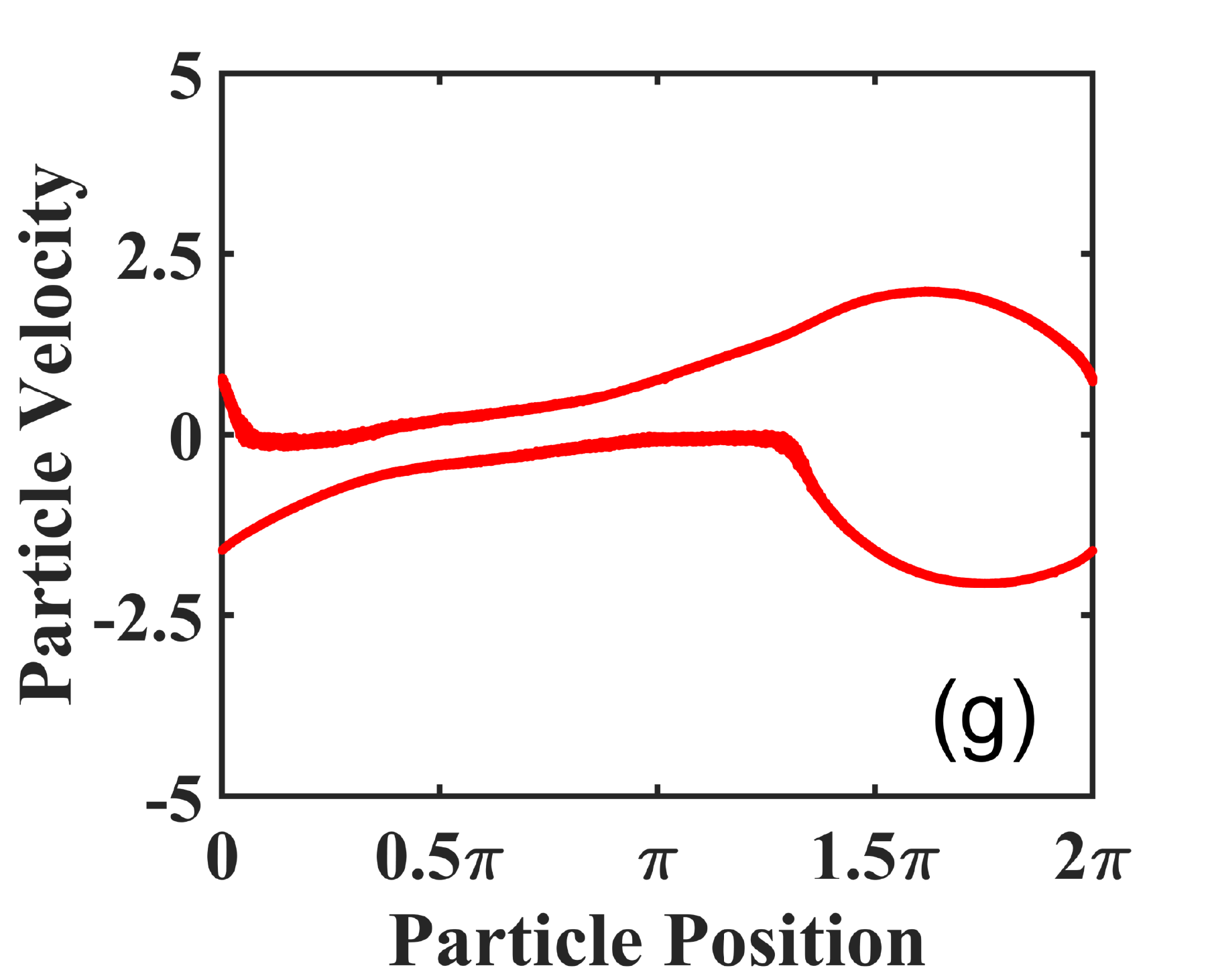}
\includegraphics[width=0.32\linewidth]{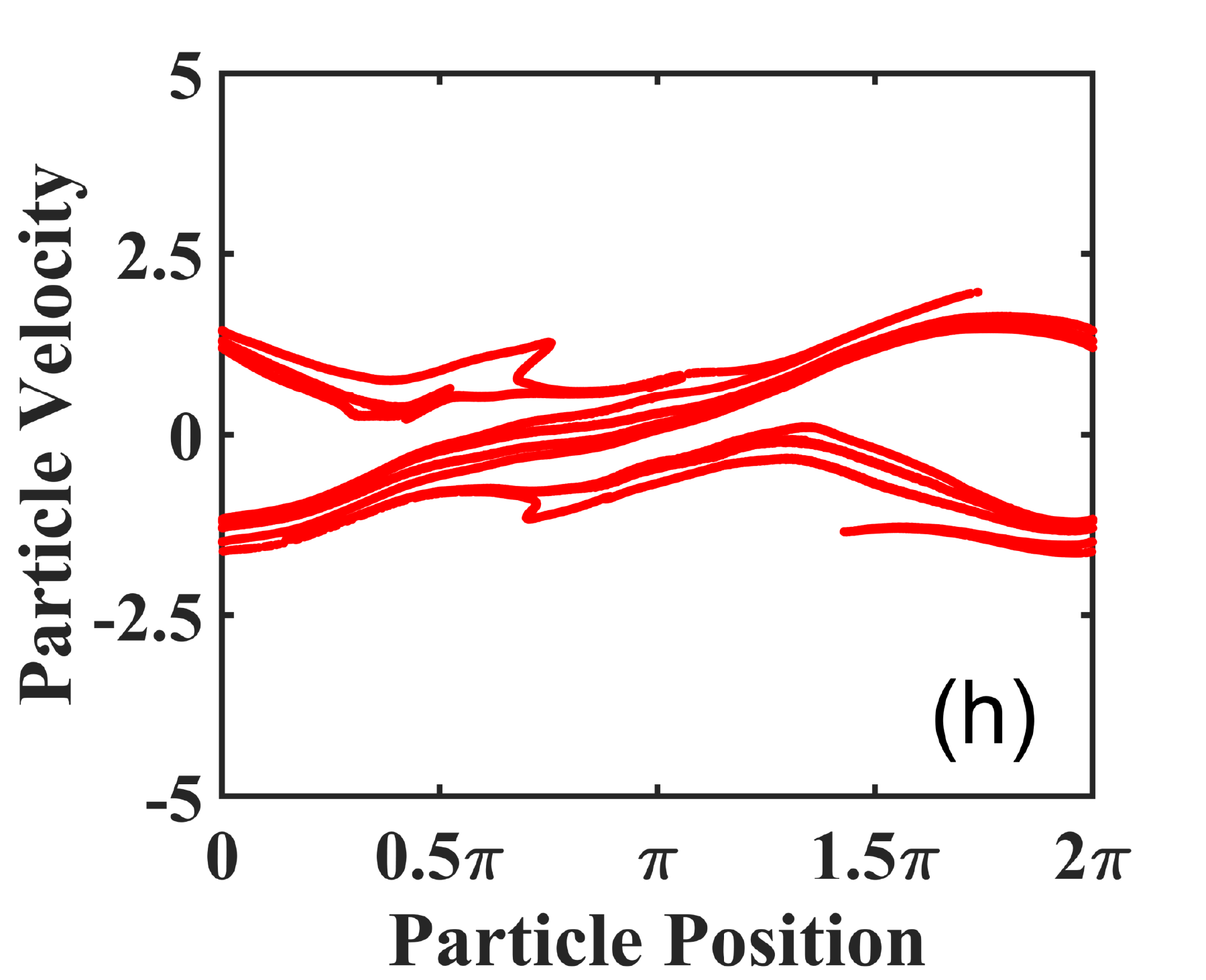}
\includegraphics[width=0.32\linewidth]{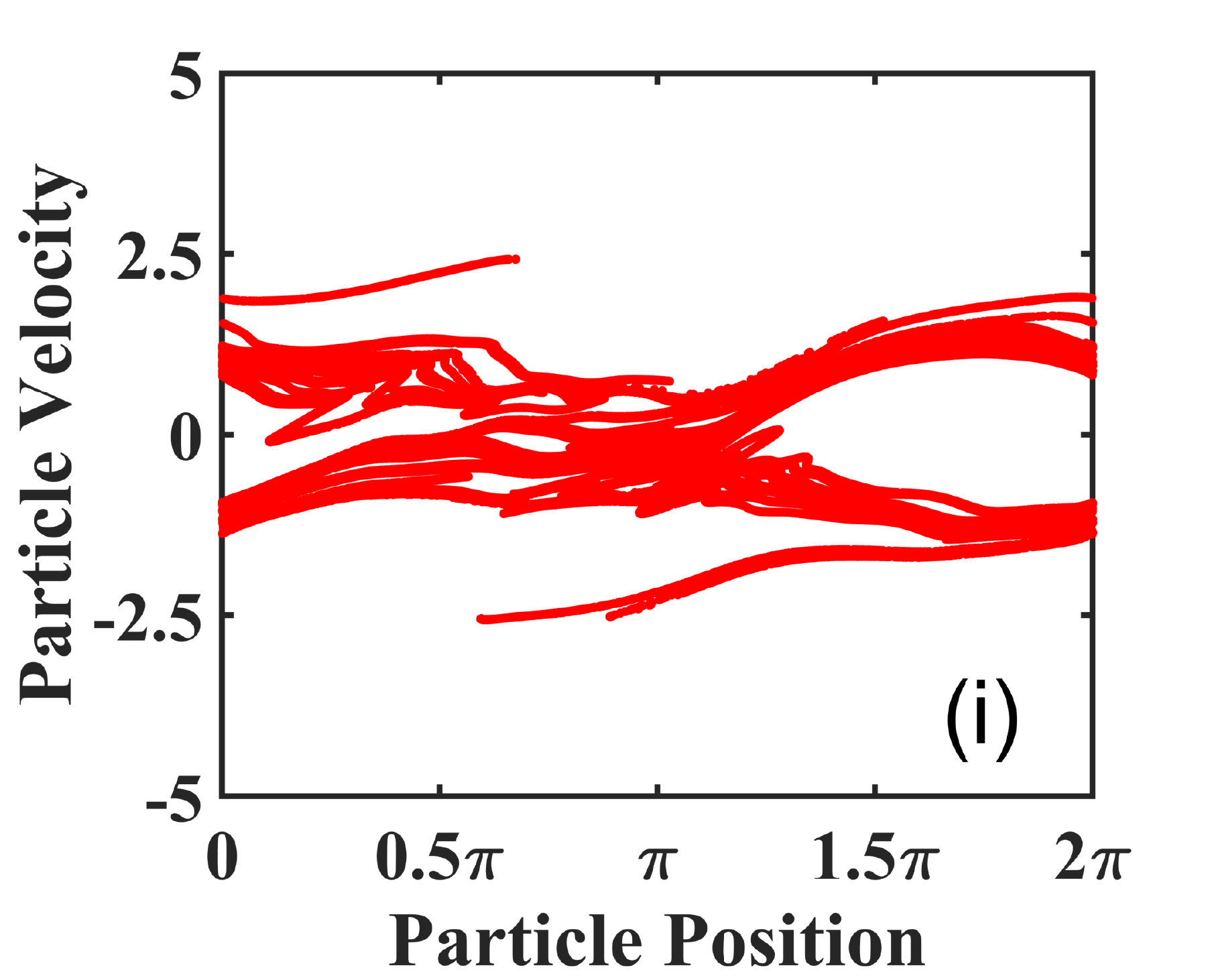}
\caption{Phase-space distributions of the macro particles at time $T=10, 15$ and $20$ for the two-stream instability. (abc) the classical explicit algorithm; (def) the AP algorithm; and (ghi) the APEC algorithm.}
\label{fig:twostream:p6}
\end{figure}

We then consider the under-resolved case by setting $\lambda =0.005$, $N_x=64$ and $\Delta t=0.1$. The other parameters are the same as the resolved case. Figures \ref{fig:twostream:p4_2} displays the results of the electric energy and the total energy evolutions with time. Figure \ref{fig:twostream:2p7} shows the APEC results for different time steps. Here, the initial total energy is $W_0=5.021160162244403$. Similar performance as the bump-on-tail instability can be observed for the three algorithms, i.e., the AP-type methods present stable simulations, but the explicit method is unstable. One can see from Figure \ref{fig:twostream:2p7}(b) that the energy charge of the APEC algorithm is at the machine precision, demonstrating the promising performance of this algorithm.

\begin{figure}[H]
\centering
\includegraphics[width=0.45\linewidth]{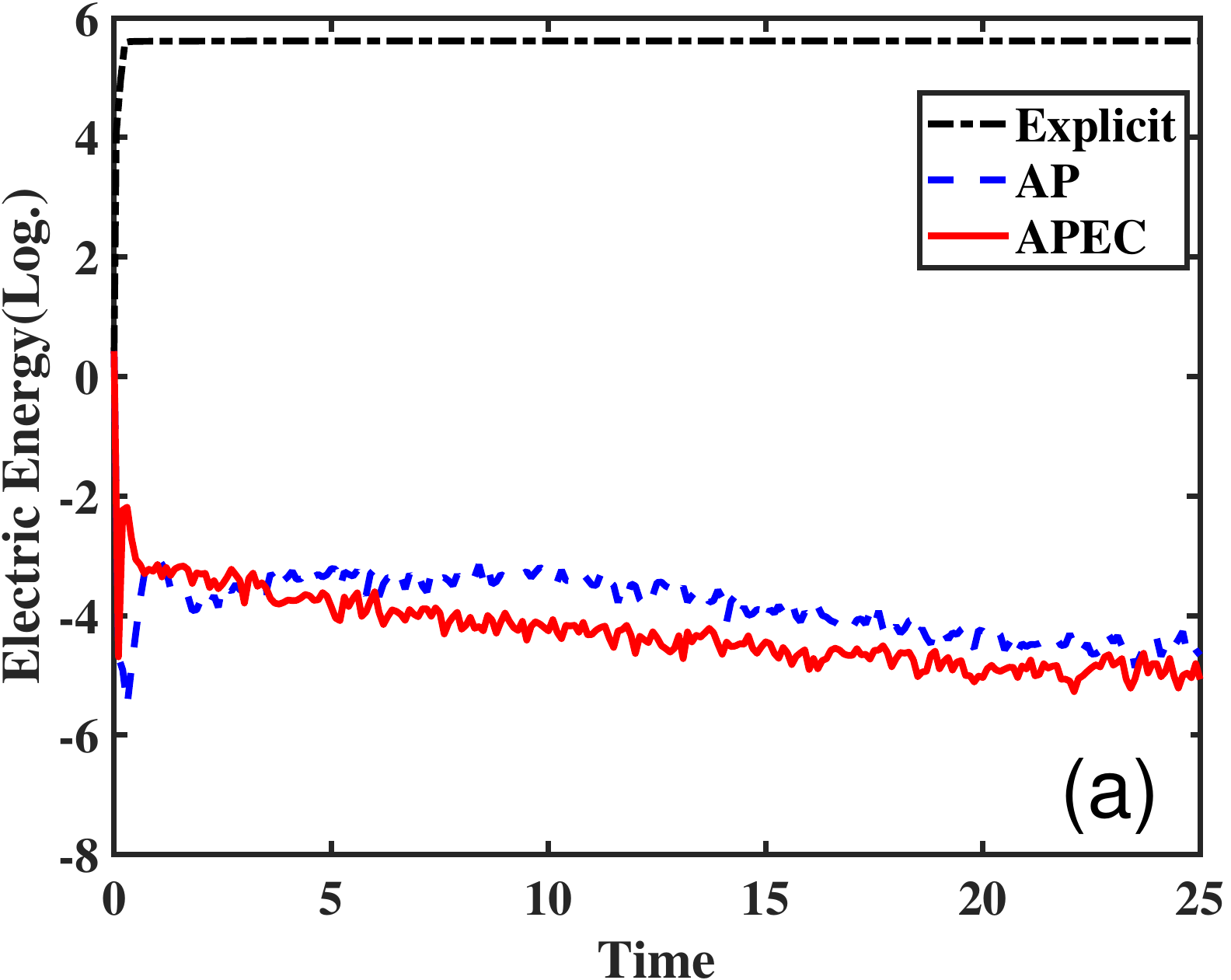}\hspace{2mm}
\includegraphics[width=0.45\linewidth]{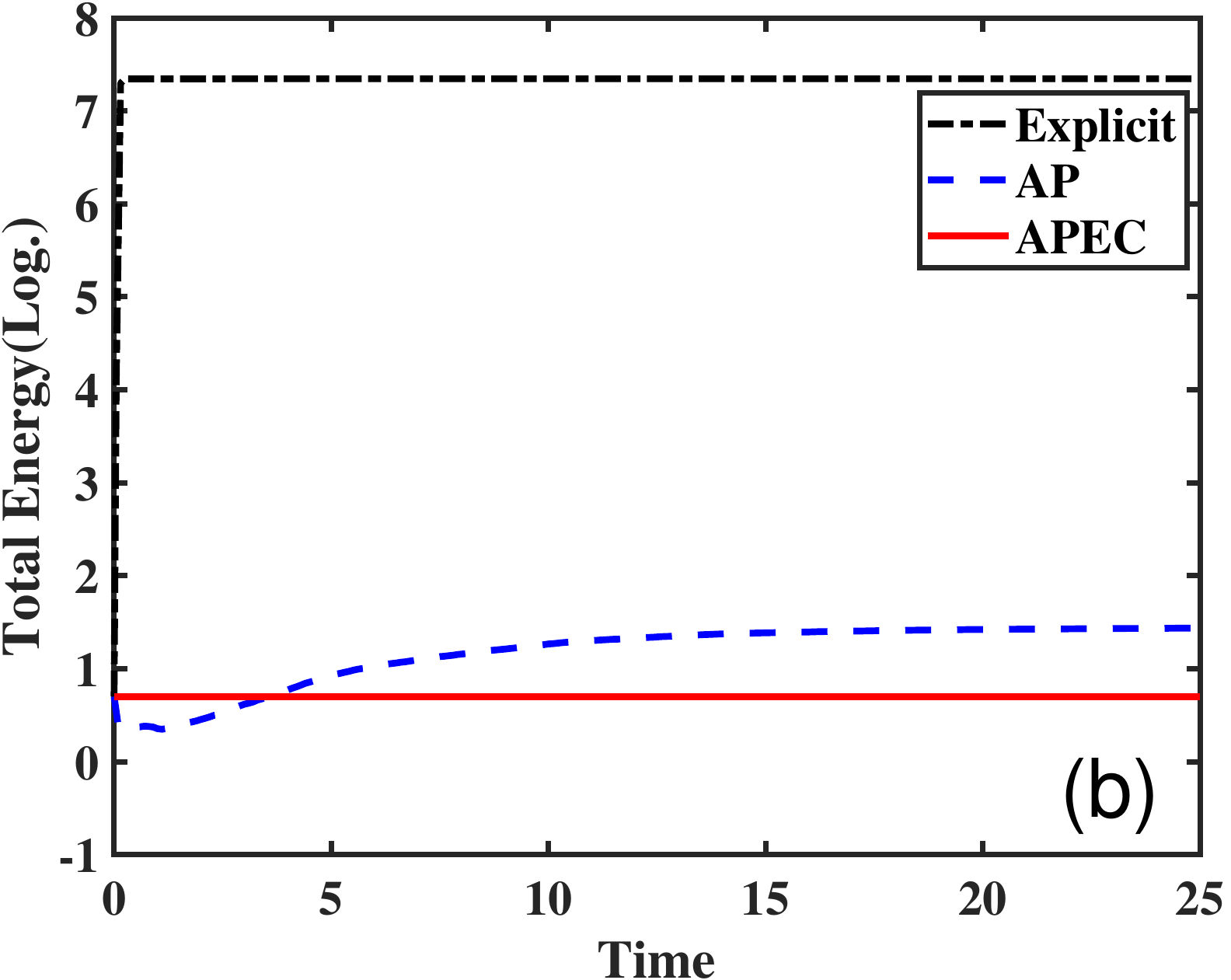}
\caption{Under-resolved case of the two-stream instability calculated by the classical explicit, the AP and the APEC schemes. (a) Electric energy with time; (b) Total energy with time.}
\label{fig:twostream:p4_2}
\end{figure}

\begin{figure}[H]
\centering
\includegraphics[width=0.45\linewidth]{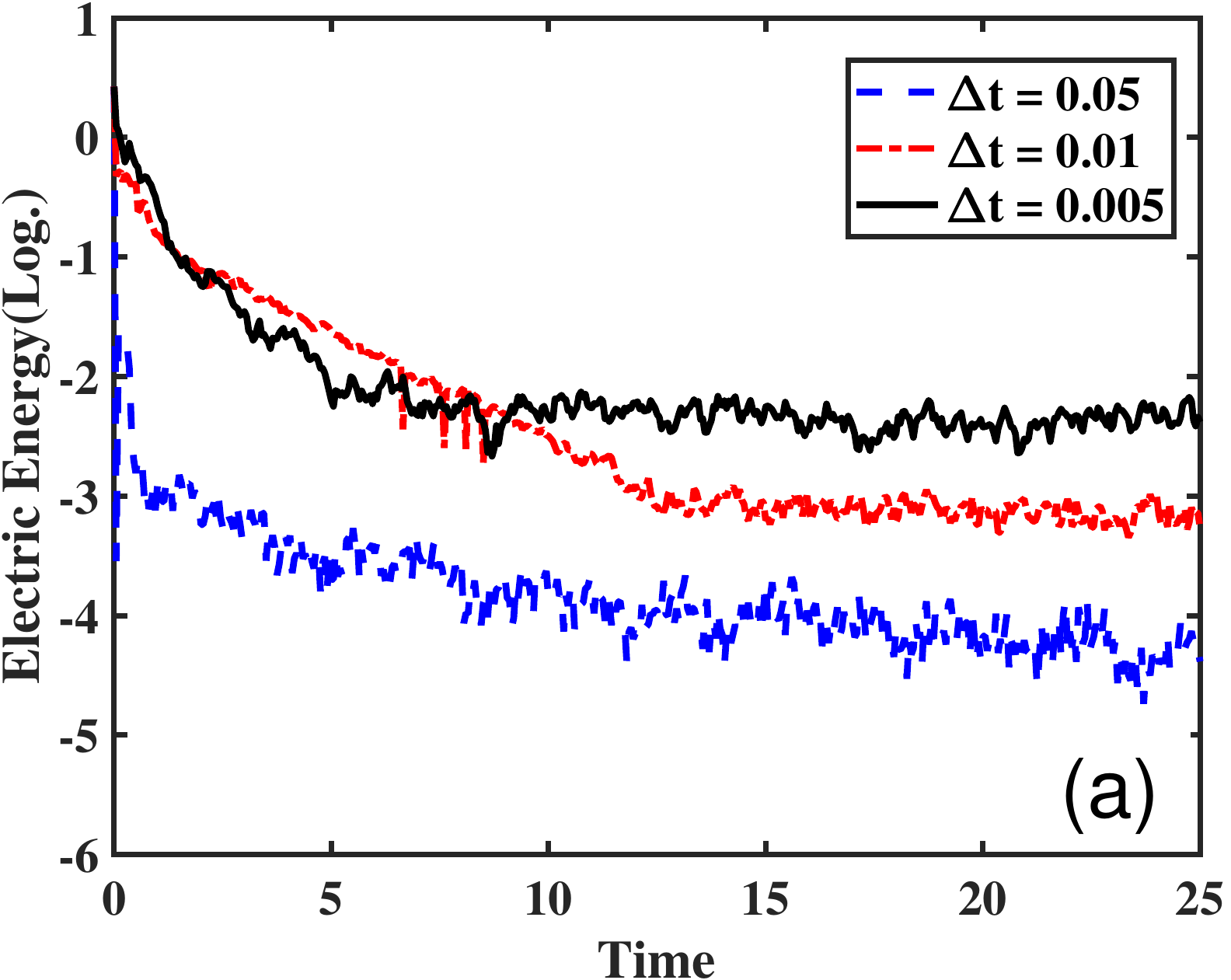}\hspace{2mm}
\includegraphics[width=0.465\linewidth]{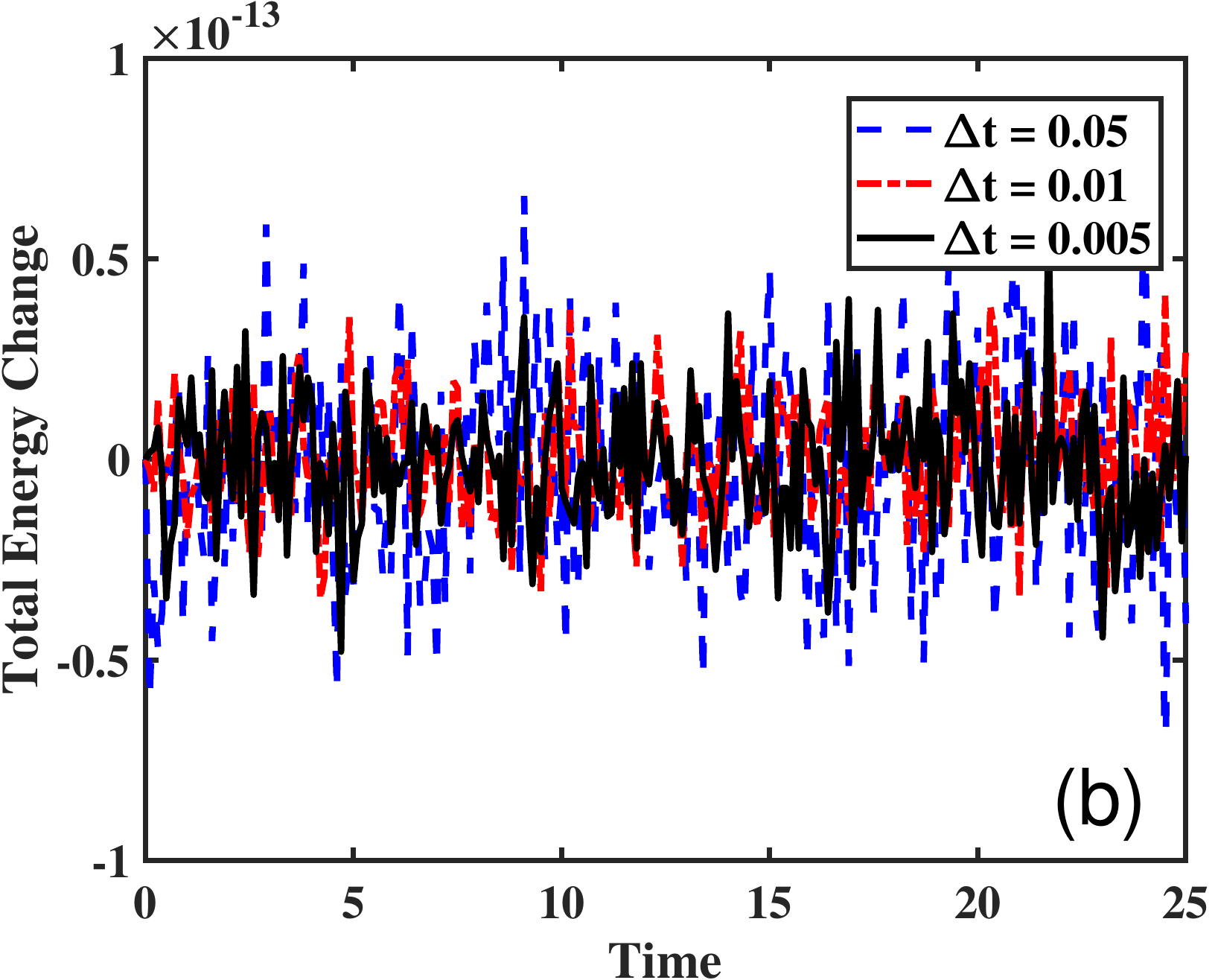}
\caption{The APEC algorithm for the two-stream instability with different time steps $\Delta t=0.005, 0.01$ and $0.05$. (a) Electric energy with time; (b) Total energy change with time.}
\label{fig:twostream:2p7}
\end{figure}

\section{Conclusion}
\label{sec:conclusion}

We have proposed an asymptotic-preserving and energy-conserving PIC algorithm for the Vlasov-Maxwell system. The algorithm can give accurate numerical solutions
on both non-neutral and quasi-neutral regimes  and the total energy is preserved exactly. In the quais-neutral regimes, the APEC method is still stable with large time steps and can capture the main mechanism
of the system while the reference explicit scheme will blow up. Several numerical tests are performed to demonstrate the attractive performance of the new algorithm.
In future, we will extend the PIC algorithm to high dimensional Vlasov-Maxwell systems for more complex plasma applications.

\section{Acknowledgement}
This work is funded by the Strategic Priority Research Program of Chinese Academy of Sciences (grant Nos. XDA25010402, XDA25010403, XDA250050500).
L. Ji acknowledges the support from China Postdoctoral Science Foundation No. 2021M702141.
Z. Yang acknowledges the support from the NSFC (No. 12101399) and the Shanghai Sailing Program (No. 21YF1421000).
Z. Xu acknowledges the support from the NSFC (grant No. 12071288).
D. Wu acknowledges the support from the NSFC (grant No. 12075204) and the Shanghai Municipal Science and Technology Key Project (No. 22JC1401500).
S. Jin acknowledges the support from the NSFC (grant No. 12031013).

\section{Data availability}
The data that support the findings of this study are available from the corresponding author upon reasonable request.

\bibliographystyle{abbrv}
\bibliography{APEC_PIC_Submission}

\end{document}